%% file: main.tex
\documentclass[11pt]{article}
\usepackage{amsfonts,amsmath,amsthm,amssymb,amstext}
\usepackage[left=0.8in,right=0.8in,top=0.8in,bottom=0.8in,bindingoffset=0.in,nohead]{geometry}
\usepackage{graphicx}
\usepackage{adjustbox}
\usepackage{subfigure}
\usepackage{url}
\usepackage{caption}
\usepackage{subcaption}
\usepackage{float}
\usepackage{booktabs}
\usepackage{blindtext}
\usepackage{longtable}
\usepackage[authoryear]{natbib}
\usepackage[figuresright]{rotating}
\usepackage{graphics}
\usepackage{setspace}
\usepackage{appendix}
\usepackage[hidelinks]{hyperref}
\usepackage{eurosym}
\usepackage{subfloat}
\usepackage{verbatim}
\usepackage{multirow,array}
\usepackage{rotating}
\usepackage{comment}
\usepackage{framed}
\usepackage[inline]{enumitem}
\usepackage{bigints}
\usepackage{makecell}%
\usepackage{lscape}
\usepackage{array}
\newcolumntype{L}[1]{>{\raggedright\let\newline\\\arraybackslash\hspace{0pt}}m{#1}}
\newcolumntype{C}[1]{>{\centering\let\newline\\\arraybackslash\hspace{0pt}}m{#1}}
\newcolumntype{R}[1]{>{\raggedleft\let\newline\\\arraybackslash\hspace{0pt}}m{#1}}

\usepackage[most,breakable,skins]{tcolorbox}
\allowdisplaybreaks

\usepackage{xcolor}

\newtcolorbox{updatebox}{
    breakable,
    colback=gray!15,
    colframe=gray!60,
    boxrule=0.5pt,
    arc=2mm,
    left=2mm,
    right=2mm,
    top=1mm,
    bottom=1mm
}

\usepackage{relsize}

\usepackage[dvipsnames]{xcolor}

\usepackage{footmisc}

\usepackage{todonotes}
\newtheorem{lemma}{Lemma}
\newtheorem{proposition}{Proposition}
\newtheorem{corollary}{Corollary}

\newtheorem{definition}{Definition}

\newtheorem{claim}{Claim}
\newtheorem*{claim*}{Claim}
\usepackage{url}

\usepackage{listings}
\usepackage{xcolor}
\usepackage{tikz}
\usetikzlibrary{patterns,arrows,decorations.pathreplacing,calc,intersections,through,backgrounds,plotmarks} 
\usepackage{pgfplots}

\usepackage{enumitem}

\usepackage{soul}

\clearpage
\newpage

\renewcommand{\thesection}{\arabic{section}}
\renewcommand{\thesubsection}{\arabic{section}.\arabic{subsection}}
\renewcommand{\thefigure}{\arabic{figure}}
\renewcommand{\thetable}{\arabic{table}}

\renewcommand{\thepage}{\arabic{page}}

\makeatletter
\def\@fnsymbol#1{\ensuremath{\ifcase#1\or \or *\or ** \or \ddagger\or
 \mathsection\or \mathparagraph\or \|\or **\or \dagger\dagger
 \or \ddagger\ddagger \else\@ctrerr\fi}}
\makeatother

\title{Artificial Intelligence in Team Dynamics:\\ Who Gets Replaced and Why?\\

\thanks{We thank Jose Apesteguia, Xiaoming Cai, Emiliano Catonini, Lester Chan, Kyle Chauvin, Yuxin Chen, Lin William Cong, Rahul Deb, 
Francesco Decarolis, Wouter Dessein, Anthony Dukes, Joseph Emmens, Hülya Eraslan, Miguel Espinosa, Tan Gan, Sanjeev Goyal, Dingwei Gu, Nima Haghpanah, Wei He, Lukas Hensel, Gaoji Hu, Shota Ichihashi, Wei Jiang, Navin Kartik, Jin Li, Xing Li, Zhuoran Lu, Chen Lyu, Daniel Martin, Kristina McElheran, Claudio Mezzetti, Weicheng Min, Mariann Ollar, Juan Andres Russy, Xiangyu Shi, Xianwen Shi, Junze Sun, Feng Tian, Venky Venkateswaran, Shaohui Wang, Ziwei Wang, Xi Weng, Qinggong Wu, Yiqing Xing, Huanxing Yang, Ming Yang, Yang You, Xiaobo Yu, Wanchang Zhang, Zemin (Zachary) Zhong, Jidong Zhou, Li-An Zhou in addition to audiences at the Marketing and Technology Conference at Columbia University, the Big Data, Artificial Intelligence, and Financial Economics Fall 2025 Conference at NBER, the Asian Conference on Organizational Economics at the University of Hong Kong, the International Conference of Marketing Science and Innovation at China University of Mining and Technology, the seminar at Fudan University, Johns Hopkins University, Paris Digital Economy Conference, Renmin University of China, University of International Business and Economics, CUHK-Shenzhen, MIT Futuretech Lab, Wuhan University, and VIDE Online Seminar series for their feedback about this project. Cheng thanks the National Natural Science Foundation of China (Grant No. 72192844) for support. Yildirim thanks the Wharton Dean's Research Fund and Mack Institute for their generous financial support for the research. Names are listed alphabetically. All errors are our own. All correspondence can be sent to the first author.}
}
\vspace{2in}

\author{Xienan Cheng\thanks{Peking University, e-mail: xienanc@umich.edu.} \and Mustafa Dogan\thanks{University of East Anglia, e-mail: m.dogan@uea.ac.uk.} \and Pinar Yildirim\thanks{Wharton School of the University of Pennsylvania and NBER, e-mail: pyild@wharton.upenn.edu.}}

\date{}

\begin{document}

\maketitle

\begin{abstract}

This study investigates the effects of artificial intelligence (AI) adoption in organizations. We ask: First, how should a principal optimally deploy limited AI resources to replace workers in a team? Second, in a sequential workflow, which workers face the highest risk of AI replacement? Third, would the principal always prefer to fully utilize all available AI resources, or are there any benefits to keeping some slack AI capacity? Fourth, what are the effects of optimal AI deployment on the wage level and intra-team wage inequality? We develop a sequential team production model in which a principal can use peer monitoring—where each worker observes a signal of their predecessor's effort—to discipline team members. The principal may replace some workers with AI agents, whose actions are not subject to moral hazard. Our analysis yields four key results. First, the optimal AI strategy stochastically replaces workers rather than fixating on a single position. Second, the principal replaces workers at the beginning and at the end of the workflow, but does not replace the middle worker, since this worker is crucial for sustaining the flow of information obtained by peer monitoring. Third, the principal may optimally underutilize available AI capacity. Fourth, the optimal AI adoption increases average wages and reduces intra-team wage inequality.
\\

\noindent {\bf Keywords:} \em Artificial intelligence (AI), automation, technology, peer monitoring, teams, organizational structure

\noindent {\bf JEL Codes:} \em L20, M21, J31, M54, O33
\end{abstract}
\thispagestyle{empty}

\setstretch{1.4}

\newpage 
\setcounter{page}{1}

\input{Introduction}
\input{Model}
\input{Extensions}

\input{Conclusion}

\bibliographystyle{apalike}
\bibliography{lit}

\setstretch{1} 
\input{Appendix}

\end{document}

%% file: Introduction.tex
\section{Introduction}

Over the past decade, artificial intelligence (AI) has profoundly changed how organizations work and how tasks are assigned to workers. Organizations worldwide are adopting AI at unprecedented scales, owing to benefits such as increased consistency of output and reduced production costs. According to a McKinsey survey, approximately 88\% of respondents reported that their organizations utilize AI ``in at least one business function'' \citep{mckinsey_state_ai_2025}. Despite its rapid adoption and projected growth, integrating AI into existing organizational structures comes with challenges \citep{deloitte_ai_2024}. Specifically, traditional workflows are designed and optimized for human workers, typically operating in teams, working on connected tasks with linked incentives. Thus, when AI is introduced into existing workflows, its effects are unlikely to remain localized. Workers whose tasks are not directly impacted by AI are still likely to be indirectly impacted. At a time when AI is reshaping work and workflows \citep{Bughin2018, Brynjolfsson2022} and organizations are thinking about how best to integrate AI into their existing systems, it is essential to ask what the implications of AI are for teams. Yet, the literature to date largely reported effects of AI focusing on individual workers, abstracting away from team dynamics.

In this study, we focus on this area and ask the following questions: How should a principal optimally deploy limited AI resources to replace workers in a team? How does the position of a worker in the production sequence impact his risk of being replaced with AI and his earnings? Would the principal always prefer to fully utilize all available AI resources, or are there any strategic benefits to keeping some slack AI capacity? What are the effects of optimal AI deployment on the intra-team wage inequality, i.e., the wage gap between the highest- and lowest-paid workers in a team?

To address these questions, we develop a model building on the framework of \cite{winter2010transparency}. 
A principal ({\em she}) manages a team of workers ({\em he}), each performing a single task for a project. 
The tasks are equally critical to the project's success and complementary to each other. 
Workers act in sequence, deciding whether to exert effort or shirk, and these choices determine the likelihood of project success. 
An AI agent ({\em it}) is available to the principal to replace a worker and perform the associated task, and we study how to deploy it optimally.
We assume effort costs are the same for AI agent and human workers, so neither resource is inherently superior to the other.  However, unlike its human counterparts, AI's actions are not subject to moral hazard, and if deployed, it always exerts effort.\footnote{The assumption that AI is not subject to moral hazard is a commonly expressed view. \cite{FGJV2025} write ``[...] the nature of the `moral hazard' differs significantly between human and AI agents.
Unlike humans, AI does not shirk in effort, become sloppy when bored or fatigued, or pursue personal perks such as corporate jets, nepotism, or empire building (the recurring themes in [...] research involving agency problems). Instead, AI rigorously optimises programmed objectives, such as profit maximisation or trading
efficiency. Nevertheless, dutiful AI brings [...] new (related) dimensions of agency
problems'' (p.115).  
Other studies also built on the moral hazard differences between human and technology resources \citep[e.g,][]{dogan2022managing, dogan2024strategic}.}

Against this backdrop, the informational environment faced by workers is such that, 
prior to making their effort decisions, workers can observe a perfect signal about their predecessor's effort but cannot tell whether it is a human or AI. 
The principal's objective is to induce all workers to exert effort at minimal compensation cost. 
As shown in \cite{winter2010transparency}, in an all-human team, the optimal compensation scheme leverages peer monitoring by incentivizing workers to shirk if their predecessor shirks, creating a ``domino effect'' that propagates shirking to all subsequent workers and thereby effectively deters shirking. 
Introducing AI into teamwork may interrupt this domino effect and alter the structure of  wages, creating trade-offs for the principal.

Specifically, there are three trade-offs that a manager needs to consider when deciding {\em whether} to replace a human worker. The first and most intuitive trade-off is the {\em direct cost savings} achieved when replacing a high-wage worker with an AI agent, fixing everything else. 
However, this benefit may be offset by other counter effects. 
A second effect, a {\em direct incentive cost}, arises when the likelihood of a worker being replaced with AI increases and the worker is not replaced. In this case, the worker infers that workers who succeed him are more likely to be replaced by AI and his temptation to shirk is higher, requiring a higher compensation to sustain his effort. 
The third tradeoff, an {\em indirect incentive cost}, occurs because a higher replacement probability for a worker also weakens the incentives of those who precede him, making their effort costlier to enforce. These effects manifest  differently for team members based on their positions in the production sequence. 

Considering the interplay of these three effects, our study yields four key insights. 
First, the principal prefers a mixed strategy in which she {\em randomizes} the deployment of AI to replace human workers.
Put differently, teams with randomized AI replacement are preferred over deterministic team compositions where some workers are replaced with certainty. 
Such randomization can be implemented in real life by probabilistically replacing workers with AI across different projects or shifts. So, when deployed optimally, AI affects human workers at the intensive margin, by altering the intensity or frequency of their work, rather than at the extensive margin, or by displacing them \citep{Jiang2025}.
This finding tempers some widespread concerns regarding labor displacement effects of AI \citep[e.g.,][]{oecd2023employment, wsj_ai_jobs_2025}.

Second, we address which team members face the highest and the lowest risk of AI replacement. For clarity, we focus on a three-worker team, where members are differentiated by position as {\em front-most}, {\em middle}, and {\em end-most}. We find that, in the optimal replacement strategy, the middle worker faces the lowest (zero) risk of replacement, so that the information flow among the members of the team can be maintained. By contrast, both the front-most and end-most workers face a positive risk of replacement, with the end-most worker being most at risk.

Third, the optimal strategy may leave some AI capacity unused. The principal may benefit strategically from keeping some slack AI resources, which generates an additional layer of uncertainty---not only about which workers are replaced, but also about whether any replacement occurs at all. This finding implies that in some cases, a manager will choose to maintain an all-human team, despite having unused AI capacity at her disposal to create a human-AI team.

Fourth, we examine how optimal AI adoption impacts workers' wages and how the impact varies by position within the team members. Without AI adoption, the wages of workers increase monotonically along the production sequence, with the end-most worker earning the highest wage. Our results show that optimal AI deployment maintains this wage hierarchy, but reduces the wage gap between team members. Specifically, the front-most and middle workers see their wages increase due to AI adoption, whereas the wage of the end-most worker (who earns the highest wage) remains unchanged. This leads to decreased intra-team wage inequality, aligning with \cite{bessen2025happens}, who find minimal wage scarring and even wage gains after the introduction of automation into an organization.

In Section \ref{sec: extensions}, we introduce several extensions to the main model. Section \ref{sec: complementarity} focuses on a critical characteristic of team production: synergy among workers \citep{che2001optimal}. We proxy synergies by the degree of complementarity between the tasks workers carry out. We propose a specific production function that parameterizes this complementarity and find that greater task complementarity enhances the principal's ability to leverage peer monitoring, reducing the wage incentives necessary to deter shirking. However, higher complementarity also implies that a worker's shirking is more detrimental, and this implies the risk of replacing the front-most worker increases and that of the end-most worker decreases.  We find that AI's potential to mitigate wage inequality is more pronounced in teams with low task complementarity.
In Section \ref{sec:TaskBased}, we focus on task-based automation \citep{acemoglu2018modeling} and introduce an alternate model where AI substitutes for a fraction of a worker's tasks, rather than substituting the worker. We find that replacing a fraction of the tasks of a worker is never optimal, and the principal prefers to replace either all or none of the tasks of a worker. Moreover, as in the main model, only the front-most and the end-most workers face risks of task replacement. 
In Section \ref{sec: network}, we examine how the optimal deployment of AI varies with the underlying network structure of the team. We compare the baseline chain network---whose optimality for human teams is established in \citet{winter2010transparency}---to a star configuration. In the star network, the central worker---the worker who facilitates the team's information flow---again faces the lowest (i.e., zero) risk of AI replacement, while that for the peripheral workers may be positive. These findings highlight how the network structure among the team members may shape the principal's AI adoption strategy. 
Section \ref{sec: strategicAI} extends our model to consider the possibility that, similar to humans, AI can be programmed to shirk. We find that, if a principal could choose, she would choose to adopt an AI with reduced efficiency: one that can shirk. 
Finally, in Section \ref{sec: AdditionalDiscussions}, we discuss the implications of relaxing principal's AI capacity constraint; allowing workers to have partial or full observability by allowing them to tell if the task of their predecessor is completed by a human or AI; and introducing heterogeneity among workers to differentiate their contributions to the project.

For practitioners considering the deployment of AI technologies and their impact on the workforce, we highlight several key strategic considerations. First, short-sighted AI replacement strategies, particularly those primarily driven by cost-cutting motives such as replacing the highest-paid workers, can backfire. A focus on direct costs of replacing workers alone overlooks the {\em indirect costs} AI introduces in the form of moral hazard. Our findings suggest that, in some cases, it may be optimal to replace a lower-paid worker if doing so minimizes the combined direct and indirect costs.
Second, we find that a stochastic AI-human teaming strategy, where work is assigned to AI on a randomized basis, can outperform deterministic assignment strategies. This implies that retaining workers while reallocating their work to AI in a stochastic manner may be more effective than full replacement of the workers. Such strategies not only preserve human capital, but may enhance productivity within teams.
Third, managers should be prepared to adjust wages for all team members, including those who are not directly impacted by AI.  The wages of all workers in a team can be impacted by AI adoption, even when a worker himself is not substituted by AI. The principal may find herself having to pay a higher wage for the remaining workers after some team members are substituted with AI, despite benefiting from some savings in the form of foregone wages. This suggests that it is imperative for managers to understand that even when AI adoption is localized (i.e., replacing a small set of workers or tasks), its spillover effects can permeate the broader team. 

Our research contributes to the literature on the role of peer monitoring in incentivizing effort within teams \citep[e.g.,][]{che2001optimal, gibbons2013chapter, villas-boas2020repeated, upton2024implicit}, as well as the literature modeling moral hazard in production settings \citep[e.g.,][]{sun2018optimal, tian2023dynamic}. We extend this body of work by investigating how the introduction of technology in general, and AI more specifically, reshapes the information structure within teams, with an emphasis on the risk of undermining peer monitoring.  We demonstrate that AI can alter the incentives and the compensation of workers. In this sense, our work also complements the recent studies that endogenize information flow within teams when designing incentive structures \citep[e.g.,][]{zhou2016economics, au2021matching, LU2025105969}.

This study is also relevant to the literature studying AI effects on worker and firm productivity. A growing number of studies experimentally estimate the effect of AI adoption on worker and business productivity \citep{dell2023super, dell2023navigating, otis2024uneven, brynjolfsson2025generative, cui2026effects, kim2026mapping}, with a particular focus on the adoption of generative AI (GenAI) tools. Overall, these studies demonstrate heterogeneous effects of AI adoption on productivity. For instance, while novice workers typically show a higher improvement in productivity compared to experienced workers  \citep{brynjolfsson2025generative}, these effects may be reversed for complex tasks \citep{dell2023navigating} or for businesses with unique needs \citep{otis2024uneven}. Workers may also strategically shift their focus to tasks that cannot be carried out by AI \citep{yiu2025strategic}. 
What is common between these studies and ours is the focus on productivity. Similar to these studies, we recognize that AI---either as a substitute or a complementary production technology---may alter a worker's effort decision. We provide a theoretical lens to explain the mixed findings from the earlier studies. In particular, our study shows two factors which may play a role: (i) the effects of AI do not remain localized---adopting AI to replace a worker may generate negative externalities on other workers; (ii) AI may be deployed sub-optimally---e.g., AI may be used indiscriminately to replace (the tasks of) the wrong workers.  Our paper highlights that, given AI's negative externalities, studies collecting data from environments where AI is deployed suboptimally may falsely conclude that AI fails to improve worker or organizational productivity.

Another stream of the literature focuses on estimating the macro effects of technology adoption, such as effects on employment and wages. A substantial body of research suggests that effects of automation on employment and wages have been nonuniform, where lower-skilled and less-educated workers have been disproportionately impacted \citep{acemoglu2019automation, acemoglu2020robots, petrova2024automation}. In contrast, the documented displacement stemming from AI thus far has been more specific to occupations and tasks \citep{Acemoglu2022ai} and impacts skilled and educated workers as well \citep{BrynjolfssonChandarChen2025_Canaries}. On wages, the evidence is similarly mixed: while some studies report declines due to automation \citep{acemoglu2020robots, petrova2024automation}, others find little to no wage scarring—or even wage gains—associated with automation and AI adoption \citep{barth2020robots, domini2022whom, bessen2025happens}. For generative AI, specifically, recent work by \cite{humlum2025large} finds no significant changes in wages and employment.

Overall, our study bridges two important strands of research: the literature on the optimal design of incentives in teams and the emerging literature on the organizational implications of AI adoption \citep[e.g.,][]{agrawal2024artificial, bastani2025human, ide2025artificial, ferreira2025artificial, demirer2026chaining, fahn2026toward}. Though most firms operate in team-based structures, much of the existing research on AI abstracts away from explicit consideration of incentives and interactions between members in teams.\footnote{A few exceptions include \cite{athey2020allocation}, \citet{dogan2022managing}, and \citet{dogan2024strategic} for automation and \cite{bastani2025human} and \citet{demirer2026chaining} on AI, which begin to explore how technological change affects organizational design and worker coordination.}  Recent experimental evidence further underscores the importance of the team perspective.
For example, \cite{dell2023super} find that AI integration can interfere with a team's ability to coordinate, and thus can decrease individual effort and overall team performance.  
These findings highlight the need for a more nuanced understanding of human-AI interactions within teams, explicitly considering the externalities that arise when some workers are replaced by AI on the workers who are not replaced. Our study addresses this gap by developing a formal model that captures both the benefits and costs of AI integration in team settings.

The rest of the paper is organized as follows. In Section~\ref{sec: model}, we set up the model, and in Section~\ref{sec: analysis} we lay out the direct and indirect costs and benefits of AI adoption and describe the optimal AI adoption strategy. We offer several extensions in Section \ref{sec: extensions}, and we provide a discussion of additional considerations and limitations. 
Finally, in Section~\ref{sec: conclusion}, we conclude with a brief discussion of our key findings and takeaways for practice. 

%% file: Model.tex
\section{Model} \label{sec: model}

A project is undertaken by a team of $n \geq 3$ workers.\footnote{Throughout the paper, we will use the terms team and network interchangeably.} Each worker (he) $i \in \{1, \ldots, n\}$ performs a single task. The manager of the organization (the principal, {\em she}) seeks to incentivize each worker to exert effort while minimizing her compensation cost. Below we describe first the production environment without AI, and then introduce the principal's option to replace a worker with an AI agent.

\paragraph{Production, information, and contracts.} Each worker $i$ makes a binary effort decision $e_i \in \{0,1\}$, where $e_i = 1$ indicates that he exerts effort and $e_i = 0$ indicates that he shirks. The cost of effort for worker $i$ is $c(e_i) = c e_i$, with $c > 0$. Workers $1, \ldots, n$ act in sequence.

Before making his own effort decision, each worker receives a signal about the contribution of his immediate predecessor. We model this signal as perfectly informative, so that each worker effectively observes whether his immediate predecessor exerted effort.\footnote{A natural interpretation is that a worker observes whether the task assigned to his immediate predecessor has been completed, since under our assumptions a task is completed if and only if the predecessor exerts effort.}\footnote{While we adopt this sequential {\em chain} structure as given, \cite{winter2010transparency} demonstrates that such a structure naturally emerges among the considered alternatives as the optimal arrangement. In Section \ref{sec: network}, we will revisit this assumption.} Worker $i$'s strategy $\sigma_i$ maps the information available to him into an effort choice:
\begin{align*}
\sigma_i &: \{0,1\} \to \{0,1\}, \text{ for each } i \in \{2,\ldots,n\}, \\
\sigma_1 &\in \{0,1\}, \text{ for worker } 1, \text{ who has no predecessor.}
\end{align*}

The project's success depends on the number of workers exerting effort: if $k$ out of $n$ workers exert effort, the project succeeds with probability $p_k$, where $p_k$ is strictly increasing in $k$. Workers' efforts are complementary, so that the marginal contribution of an additional worker's effort to the success probability rises with the number already exerting effort. Formally, $p_{k+2} - p_{k+1} > p_{k+1} - p_k$ for each $k \leq n-2$.

The principal observes only the outcome of the project, neither the workers' individual effort choices nor the signals they receive. She offers a contract $w \equiv (w_1, \ldots, w_n)$, where $w_i$ specifies the payment to worker $i$ if the project succeeds. Workers are risk-neutral and maximize their expected compensation net of the cost of effort.\footnote{Because workers have limited liability, the optimal contract grants each worker a positive payment if the project succeeds and no payment if it fails.}

\paragraph{AI replacement.} The principal has access to an AI agent that can carry out the task of a worker. The AI always exerts effort, and its effort is just as effective as a worker's in helping the project succeed. The AI's cost of effort is also identical to that of a human, $c$, and is incurred by the principal.\footnote{We purposely keep the effort costs of the AI and the worker identical. This allows us to focus on moral hazard and peer monitoring as drivers of technology adoption, rather than cost reduction.}

The principal decides (i) whether to replace a worker with AI, and if so, (ii) which of the $n$ workers to replace. We allow her to use a mixed strategy for both decisions. Her {\em replacement strategy} is
\begin{align*}
x \equiv (x_1, x_2, \ldots, x_n), \quad \bar{x} \equiv \sum_{i=1}^n x_i \leq 1,
\end{align*}
where $x_i$ is the probability of replacing worker $i$ with AI, and $\bar{x}$ is the aggregate probability of deploying AI. The constraint $\bar{x} \leq 1$ treats AI as a limited resource, capturing the financial constraints managers face in reality, which force them to prioritize which worker to replace. As we will show, the principal may optimally set $\bar{x} < 1$ and leave some AI capacity unused.

The principal commits to her replacement strategy $x$ publicly. 
Workers, however, do not observe the realized outcome of $x$: they cannot tell whether any (or which) workers have been replaced. The signal each worker receives about his immediate predecessor reflects the predecessor's contribution to the project, but is silent about whether that contribution originated from a human or from an AI, since the two are equally productive. A worker who receives a signal indicating no contribution can therefore infer that his immediate predecessor is human and has shirked.

\paragraph{Timing.} First, the principal commits to a replacement strategy and a compensation scheme $(x, w)$. Next, workers act in sequence: each observes $(x, w)$ and, where applicable, the signal about his immediate predecessor's effort, and decides whether to exert effort. The project outcome is then realized, and workers are paid according to $w$.

\subsection{Discussion of Model Setup}

The introduction of AI into the production environment requires several modeling choices that warrant discussion. We clarify the reasoning behind each below.

\paragraph{AI does not shirk.} The key distinction between AI and human workers in our setting lies in the risk of moral hazard: humans may shirk, whereas AI never does. We intentionally abstract from any other differences between the two production sources. The absence of shirking is therefore a feature of the model rather than a limitation, and empirical evidence supports this assumption: scholars recognize AI's consistent effort as an important factor in adoption decisions \citep{FGJV2025}. One may nonetheless ask whether AI could be programmed to condition its action on that of its predecessor, much as humans do. We explore this possibility formally in Section \ref{sec: strategicAI}.

\paragraph{Full versus partial replacement of worker tasks.} \ 
Each worker in our model carries out a single, indivisible task, so deploying AI in place of a worker amounts to fully replacing him. In some applications, however, AI may only partially replace a worker by handling specific components of his task. We consider such a partial replacement model in Section \ref{sec:TaskBased}.

\paragraph{Limited capacity for AI replacement.} We assume the principal has limited AI resources, requiring her to prioritize which workers to replace. This assumption can be micro-founded by utilization costs for AI. In practice, such costs are a major obstacle to scaling: executives frequently cite expenses for computing and cloud infrastructure as barriers to expanding AI capacity \citep{ibm_The_hidden_costs}. A recent survey reports that every executive interviewed had canceled or postponed at least one GenAI initiative due to compute costs, with 15\% of projects on hold and 21\% failing to scale for this reason \citep{ibm_ai_economics_compute_cost}. Managers thus face real trade-offs in deploying limited AI capacity. We relax this assumption in Section \ref{sec:LimitedCapacity}; in addition, the task-based extension in Section \ref{sec:TaskBased} accommodates more flexible capacity constraints.

\paragraph{Imperfect observability.} 
Our model assumes that the signal a worker receives about his predecessor's contribution is silent about whether that contribution came from a human or from an AI. This is the relevant case in environments where tasks are temporally separated or workers are not physically proximate, and where the output of a completed task passes a Turing test in the sense that it cannot be reliably attributed to either source. In Section \ref{sec:ImperfectObservability}, we explore the alternative case where workers can observe whether their coworkers are AI or human. The analysis shows that this scenario yields a weaker outcome for the principal compared to our baseline, suggesting that concealing whether a task is performed by a human or AI benefits her.   

\section{Analysis} \label{sec: analysis}
Our objective is to characterize the optimal AI replacement strategy and compensation scheme as defined below.

\begin{definition}
A replacement strategy and compensation scheme pair $(x, w)$ is optimal if
\begin{itemize}
 \item[(i)] it induces an equilibrium in which all workers exert effort, i.e., $(e_1, \ldots, e_n) = (1, \ldots, 1)$ on the equilibrium path;
 \item[(ii)] among all pairs inducing such an equilibrium, it minimizes the expected cost for the principal:
\begin{equation}
W_{x,w}\equiv \sum\limits_{i=1}^n \left[ x_i c + (1-x_i)p_n w_i\right]. 
\label{Wxw}
\end{equation}
\end{itemize} 
\end{definition}

\noindent Two factors shape the expected compensation cost: the principal bears the AI's cost $c$ whenever AI is deployed, regardless of the project outcome; and each worker who is not replaced is compensated only upon the project's success, which occurs with probability $p_n$ when all workers exert effort.

\paragraph{Optimal compensation for a fixed replacement strategy.} We begin by fixing the replacement strategy $x$ and characterizing the optimal compensation scheme conditional on $x$, which we denote by $w^x$. The following proposition characterizes $w^x$ and shows that it induces a trigger strategy profile $\sigma^* = (\sigma^*_1, \ldots, \sigma^*_n)$ as an equilibrium, where each worker exerts effort as long as his signal indicates that his immediate predecessor exerted effort, and shirks otherwise. Formally,
\[
\sigma_1^* = 1, \quad 
\sigma_i^*(e_{i-1}) = 
\begin{cases}
1 & \text{if } e_{i-1} = 1, \\ 
0 & \text{if } e_{i-1} = 0,
\end{cases}
\quad \text{for all } i \in \{2, \ldots, n\}.
\]

\medskip 
\begin{proposition}\label{prop: compensation}
Fix a replacement strategy $x = (x_1, \ldots, x_n)$. The optimal compensation scheme conditional on $x$, denoted by $w^x = (w_1^x, \ldots, w_n^x)$, satisfies the following for each worker $i \in \{1, \ldots, n\}$ with $x_i<1$:
\begin{align}
w^x_i = \frac{c}{ p_n -\zeta_i^x}, \label{eq: w}
\end{align}
where $\zeta_i^x$ is the resulting success rate when worker $i$ shirks: 
\begin{align} \zeta_i^x= p_{i-1} \sum\limits_{k=1}^{i-1} \frac{x_k}{1 - x_i} \, + \, \sum\limits_{k=i+1}^n \frac{x_k}{1 - x_i} p_{n-k+i} \, + \, p_{i-1} \frac{1-\bar{x}}{1-x_i}.
\label{eq: Ai} 
\end{align} 
\end{proposition}
\medskip 

Proposition~\ref{prop: compensation} characterizes the optimal compensation for each worker who is not fully replaced ($x_i<1$); if he is fully replaced ($x_i=1$), his compensation is irrelevant. We denote by $W_x$ the principal's expected compensation cost under the {\em optimal} scheme $w^x$, obtained by substituting $w = w^x$ into \eqref{Wxw}:
\begin{align}
W_x \equiv W_{x, w^x}=\sum\limits_{i=1}^n \left[ x_i c + (1-x_i)p_n w_i^x\right]. \label{eqn:expectedcompensationcost}  
\end{align}

The proof, given in the appendix, proceeds in three steps: $w^x$ supports the trigger strategy profile $\sigma^*$ as an equilibrium; no cheaper scheme can do so; and no alternative strategy profile that achieves joint effort can be supported at a lower cost.

In the trigger strategy equilibrium induced by $w^x$, each worker is indifferent between exerting effort and shirking when his signal indicates that his immediate predecessor exerted effort. Effort complementarity then makes shirking the optimal response when the signal indicates that the predecessor shirked. This generates a \emph{domino effect}: following a deviation, every subsequent human worker also shirks, and shirking propagates down the chain until it reaches an AI agent, if any. Because the AI exerts effort unconditionally, the human worker placed immediately after an AI receives a signal of contribution and reverts to effort, so the cascade stops at that position. The domino is the strongest possible deterrent to deviations and allows the principal to fully leverage peer monitoring; replacing a worker---and especially a middle worker---interrupts it and alters the incentives and compensation of the remaining workers. The principal accounts for this disruption when choosing her replacement strategy, as we shall see shortly.

Worker $i$'s indifference condition on the equilibrium path, which pins down $w_i^x$, is:
\[
p_n w_i^x - c = \zeta_i^x w^x_i,
\]
where $\zeta_i^x$ is the success rate from worker $i$'s perspective when he shirks, given that all other workers follow the trigger strategy. The payment worker $i$ receives upon project success (if not replaced) is therefore $w_i^x = c/(p_n - \zeta_i^x)$. A higher $\zeta_i^x$ makes shirking less consequential for success and thus more tempting, so the principal must offer a larger payment to offset this temptation.

Since $\zeta_i^x$ depends on worker $i$'s position in the production sequence, so does his compensation. In the absence of AI, as noted by \cite{winter2010transparency}, compensation rises monotonically along the chain, with the end-most worker paid the most. The reason is that under the trigger strategy, the end-most worker's shirking results in the smallest decline in the success rate, since it does not induce any subsequent shirking.

To understand how a worker's incentive to shirk depends on the AI strategy, we examine $\zeta_i^x$ and its components more closely. As shown in equation~\eqref{eq: Ai}, $\zeta_i^x$ is a weighted sum of three success probabilities, each corresponding to a different replacement outcome when worker $i$ deviates:
\begin{itemize}
\item If a predecessor of worker $i$ is replaced, all of $i$'s successors shirk, and the project succeeds with probability $p_{i-1}$.
\item If a successor of worker $i$, say worker $k$, is replaced, worker $i$ and his successors up to $k-1$ shirk, while the AI at position $k$ and all later workers exert effort, yielding success probability $p_{n-k+i}$.
\item If no worker is replaced, all of $i$'s successors shirk, resulting in success probability $p_{i-1}$.
\end{itemize}
The weights reflect worker $i$'s belief about replacement outcomes conditional on not being replaced himself: probability $\frac{x_k}{1 - x_i}$ on each $k \ne i$ being replaced, and probability $\frac{1 - \bar{x}}{1 - x_i}$ on no replacement at all.

Rewriting equation~\eqref{eq: Ai} highlights its dependence on the replacement strategy:
\begin{align}
\zeta_i^x = p_{i-1} + \sum\limits_{k=i+1}^n \frac{x_k}{1 - x_i} \left(p_{n-k+i} - p_{i-1}\right). \label{eq: zeta}
\end{align}
For any worker $i < n$, $\zeta_i^x$ increases with (i) the likelihood of replacing worker $i$ himself and (ii) the likelihood of replacing any of his successors. For the end-most worker, $\zeta_n^x$ is independent of the replacement strategy $x$, since he has no successors whose actions could be redirected by his shirking.

\paragraph{Tradeoffs of replacement.} We analyze the tradeoffs of replacing a worker with AI by examining the partial derivative of the principal's expected cost with respect to $x_i$:
\begin{align}
\frac{\partial W_{x}}{\partial x_i} = -\underbrace{(p_n w_i^x -c)}_{\substack{\text{direct}\\ \text{cost saving}}} + \underbrace{(1-x_i)p_n\frac{\partial w_i^x}{ \partial x_i}}_{\substack{\text{direct}\\\text{ incentive cost}}} + \underbrace{\sum\limits_{k=1}^{i-1}(1-x_k)p_n\frac{\partial w_k^x}{\partial x_i}}_{\substack{\text{indirect}\\\text{ incentive cost}}}.   \label{eq: dWdx}
\end{align}
Equation~\eqref{eq: dWdx} shows that the principal balances three distinct effects when making the optimal replacement decision.

The first effect is the {\em direct cost saving}. The expected compensation for worker $i$ is $p_n w_i^x$, while the cost of AI is $c$. Since $p_n w_i^x \geq c$, replacing a worker reduces the principal's expected compensation cost. This benefit is largest for the end-most worker, who commands the highest wage.

The second effect is the {\em direct incentive cost}. As noted following equation~\eqref{eq: zeta}, increasing the likelihood of replacing worker $i$ amplifies his incentive to shirk in the event he is not replaced. The reason is that a higher $x_i$ raises his posterior belief that some other worker has been replaced; conditional on not being replaced himself, his own shirking thus appears less consequential, and the principal must offer higher compensation to preserve his effort.

The third effect is the {\em indirect incentive cost}, which arises through changes in the compensation paid to other workers. Raising the probability of replacing worker $i$ weakens the incentives of his predecessors as well, since shirking becomes less consequential for them. The magnitude of this cost depends on worker $i$'s position: the closer he is to the end of the chain, the more predecessors are affected, and the larger the number of wage adjustments.

The principal needs to balance these three effects when designing the optimal replacement strategy. For instance, replacing the end-most (and highest-paid) worker yields the largest direct cost saving, but also generates the greatest indirect incentive cost, since all predecessors' incentives are weakened. At the other extreme, replacing the front-most worker creates no indirect incentive cost, as he has no predecessors, but the direct cost saving is smaller because his compensation is relatively low. The direct incentive cost is zero for the end-most worker, while for other workers it depends on both their position and the AI replacement strategy. In what follows, we analyze how these trade-offs between direct cost savings, direct incentive costs, and indirect incentive costs shape the principal's optimal replacement policy.

 \subsection{Optimal Replacement Strategy}

In this section, we characterize the optimal strategy for replacing workers with AI, denoted by $x^*$. As an initial step, the following proposition and the subsequent discussion highlight a key property of this strategy: randomization. We also explore the trade-offs that the principal faces while deciding her AI strategy.

\medskip
\begin{proposition}\label{prop: randomization}
The principal's optimal AI adoption strategy necessarily involves randomization: 
\begin{itemize}
\item no worker is replaced with certainty, i.e., $x_i^* < 1$ for all $i$; and
\item at least one worker is replaced with positive probability, i.e., $x_i^* > 0$ for some $i$.
\end{itemize}
\end{proposition}
\medskip

Proposition~\ref{prop: randomization} shows that the principal adopts a randomized replacement strategy, ruling out both full replacement of any worker ($x_i = 1$ for some $i$) and no replacement at all ($x_i = 0$ for all $i$). We explain this finding by highlighting the trade-offs in the principal's replacement strategy, which encompasses two critical aspects: {\em whether} to replace any worker, and, if so, {\em which} worker(s) to replace.

To see why randomization is optimal, first consider a pure replacement strategy, where the principal replaces one worker with certainty or does not replace any worker at all. Within the set of pure strategies, replacing either the front-most or the end-most worker is among the optimal decisions.\footnote{This is formally stated and proved in \ref{sec:AppendixA} (Lemma~\ref{topbottom}).} Both options result in the same expected compensation, as they create an identical team network structure: $n-1$ sequentially connected workers and one isolated AI agent. This configuration preserves information flow among the non-replaced workers, allowing the principal to fully leverage peer monitoring. Replacing a middle worker, however, results in a disconnected network and  a disrupted flow of information, increasing the compensation cost incurred by the principal. 

Now consider a class of replacement strategies where the principal replaces the front-most and end-most workers with probabilities $\rho$ and $1 - \rho$, respectively, for some $\rho \in [0, 1]$. The extreme points of this class, with $\rho = 0$ and $\rho = 1$, correspond to the two optimal pure replacement strategies. We show that all interior values of $\rho \in (0, 1)$ outperform these pure strategies. Hence, the optimal replacement strategy necessarily involves randomization.

To explain why an interior value of $\rho$ outperforms the optimal pure strategies, consider decreasing $\rho$, or gradually shifting the replacement probability from the front-most to the end-most worker.  
Within this family of strategies, the end-most worker is compensated more than the front-most worker---specifically, $\zeta_1 = p_1$ and $\zeta_n = p_{n-1}$, so $w_n > w_1$ whenever $n \geq 3$---and shifting the replacement probability toward him results in higher direct cost savings for the principal.\footnote{The wage ranking here differs from the all-human benchmark, where $\zeta_1 = p_0$ and $\zeta_n = p_{n-1}$. The reason is that under this family of strategies, the front-most worker, when not replaced, knows that the end-most worker has been replaced with certainty, which raises the success rate following his own shirking from $p_0$ to $p_1$.} However, since the compensation of both workers is independent of $\rho$, the magnitude of these savings remains constant as $\rho$ varies. This is because, in this class of strategies, the principal fully exhausts the AI capacity, and the end-most (front-most) worker knows that the front-most (end-most) worker is replaced with certainty when he himself is not replaced. As a result, $\zeta_1$ and $\zeta_n$ (and therefore $w_1$ and $w_n$) remain fixed as $\rho$ varies. 

Yet, this shift also increases the indirect incentive cost. Specifically, as the replacement probability of the end-most worker ($1 - \rho$) increases, shirking becomes less consequential for all his predecessors (except the front-most worker), prompting the principal to raise their compensation. Importantly, this indirect cost is convex in $\rho$: for each intermediate worker $i \in \{2, \ldots, n-1\}$, the shirking success rate $\zeta_i$ is linear in $\rho$, while the wage $w_i = c/(p_n - \zeta_i)$ is convex in $\zeta_i$, so the composition $w_i$ is convex in $\rho$. The following observations together imply that an interior $\rho \in (0, 1)$ is optimal within this class of strategies: (i) both extreme values $\rho = 1$ and $\rho = 0$ yield the same compensation cost, (ii) reducing $\rho$ increases the direct cost savings at a constant rate, and (iii) reducing $\rho$ increases the indirect incentive cost at an increasing rate. Therefore, all randomized replacement strategies with an interior $\rho$ dominate the pure strategies.

\paragraph{Discussion of randomization.} A few additional points related to Proposition~\ref{prop: randomization} are worth noting. The first concerns the optimality of randomization over pure replacement decisions, and whether implementing such a strategy may pose challenges in practice. Since human resources are not easily fungible, randomization can be achieved in other ways. A common approach involves randomizing the assignment of AI across projects, so that a given task is carried out by an AI some of the time and by a worker at other times. A randomization schedule can also be implemented when workers rotate across different divisions of the organization, with AI replacing them as they rotate out of a division, or by rotating workers across different work shifts (e.g., day/night) or across days of work. Companies like Dmall specialize in offering such technology solutions to rotate workers with AI across shifts. One application, for instance, staffs grocery stores with human security personnel during the day and replaces them with AI surveillance agents at night.

An implication of randomization is the emergence of hybrid (human-AI) teams, where the same task can be carried out by either a human worker or an AI. In this outcome, workers are not fully displaced but instead see a reduction in their workload, as they are occasionally substituted with AI. While this finding is in line with earlier studies documenting labor displacement effects of new technologies \citep[e.g.,][]{acemoglu2020robots, chen2025displacement}, it also offers a more mitigated and perhaps more optimistic perspective. Despite concerns about sweeping labor displacement effects from AI adoption \citep{runciman_end_2023, oecd2023employment}, our findings suggest that permanently replacing workers with AI may be suboptimal and that varying their participation through AI (e.g., across different projects or work shifts) may be more profitable for an organization.

\paragraph{Heterogeneous replacement risk based on worker position.} In presenting the next set of results, to deliver sharper outcomes, we focus on the case where $n=3$. This three-worker structure captures the distinct incentives of workers within a team:
\begin{itemize}
  \item the {\em front-most} worker (worker 1) has no predecessor, but his contribution generates a signal for his successor;
  \item the {\em middle} worker (worker 2) receives a signal about his predecessor and also generates a signal for his successor, thereby acting as a connector within the team; and 
  \item the {\em end-most} worker (worker 3) receives a signal about his predecessor but has no successor, so his effort choice is less consequential for the project's success.
\end{itemize}
With this structure in mind, we proceed to discuss key properties of the optimal randomization scheme outlined in Proposition~\ref{prop: randomization}.

\medskip
\begin{proposition} \label{prop: x2=0}
In the optimal replacement strategy, the middle worker is never replaced with AI. The end-most worker is replaced with a strictly positive probability, and this probability is weakly higher than that of the front-most worker. Formally, $x_2^*=0$, $x_3^*>0$, and $x_3^* \geq x_1^* \geq x_2^*$. 
\end{proposition}
\medskip

Proposition~\ref{prop: x2=0} highlights that workers' risk of being replaced by AI varies by their position in the production sequence. The intuition we present here is heuristic; the formal argument is given in the appendix.

First, the middle worker, who is essential for the connectivity of the signal chain, does not face the risk of AI replacement. To see this, consider an arbitrary replacement strategy $(x_1, x_2, x_3)$ in which the middle worker is replaced with positive probability, i.e., $x_2>0$. Then consider deviating from this strategy by reallocating all AI resources from the middle to the end-most worker, yielding $(x_1, 0, x_3 + x_2)$. This shift (i) increases the direct cost savings, since the end-most worker has the highest wage, and (ii) reduces the direct incentive costs, since replacing the end-most worker incurs none. It also affects the indirect incentive costs through the wages of the end-most worker's large number of predecessors. The combined positive effects of (i) and (ii) outweigh the change in indirect costs, making the principal better off; the formal argument tracking these effects in detail is in the appendix.

Second, the end-most worker's risk of replacement is strictly positive, and while the front-most worker also faces a risk of replacement, this risk is lower than that of the end-most worker. To see why, consider deviating from a given replacement strategy where $x_1 = x_3 + \Delta$ by reallocating an additional $\Delta$ amount of resources from the front-most to the end-most worker, yielding $(x_1-\Delta, x_2, x_3 + \Delta)$. This reallocation generates the same combination of effects (i), (ii), and a change in indirect costs analogous to the previous paragraph, and once again the positive effects from (i) and (ii) outweigh the change in indirect costs. Thus, any strategy where the front-most worker is replaced at a greater rate than the end-most worker cannot be optimal. This, together with the optimality of randomization (Proposition~\ref{prop: randomization}), implies that the end-most worker's replacement probability is always strictly positive, i.e., $x_3^*>0$.

From a practical perspective, these findings imply that managers wishing to integrate AI optimally into their existing organizational structures need to take into account factors beyond costs and technological feasibility. It is essential that managers consider how AI integration will disrupt the signal chain among workers in a team and follow an AI strategy that preserves it. In our setting, this is equivalent to not replacing the middle agent. Moreover, AI adoption strategy should go beyond naive replacement strategies---those that focus on replacing high-compensation workers---and also focus on reducing the negative externalities of replacement on other workers. In our setting, this corresponds to the possibility of replacing the front-most worker, rather than focusing on the end-most worker alone.

\paragraph{Utilization of AI capacity.} So far, we have described some key aspects of the optimal AI replacement strategy in an organization, assuming the principal has sufficient resources that she can optimally allocate. An important question is whether the principal would always choose to exhaust these resources.

If AI adoption decisions were purely driven by direct cost savings, the principal would always exhaust her AI resources. But in Proposition~\ref{prop: useup}, we find a seemingly counterintuitive result: the principal may choose not to fully utilize the AI capacity available to her. Replacing workers with technology is therefore not always the most preferred or the most profitable managerial strategy, and underutilization of AI resources is possible.

\medskip
\begin{proposition} \label{prop: useup} 
The principal chooses to underutilize AI capacity, setting $\bar{x}^* < 1$, if and only if the condition $p_1^2 - p_3 \, p_0 > 0$ is satisfied.
\end{proposition}
\medskip

Proposition~\ref{prop: useup} establishes that the principal does not always fully utilize the available AI capacity, and it provides a necessary and sufficient condition for when underutilization is optimal. Given that we have already established the absence of middle replacement ($x_2 = 0$), the key implication is that, despite having the option to increase the replacement probability of the front-most and the end-most workers, the principal may choose not to do so. The condition $p_1^2 - p_3 p_0 > 0$ is equivalent to $\frac{p_1}{p_0} > \frac{p_3}{p_1}$, meaning that the proportional gain from securing the first unit of effort (moving from $p_0$ to $p_1$) exceeds the proportional gain from adding the remaining units (moving from $p_1$ to $p_3$). Put differently, when the condition is satisfied, the first unit of effort is pivotal: it contributes disproportionately more to success. Under full utilization ($\bar{x} = 1$), an AI agent is always deployed, so this pivotal unit of effort is exerted regardless of the workers' actions. This raises the success probability following any worker $i$'s shirking ($\zeta_i$), making shirking less consequential and thus more tempting. Stronger incentives---i.e., higher payments---are then required to deter shirking, and to avoid this, the principal may optimally underutilize AI capacity ($\bar{x} < 1$).

To see what drives the underutilization result in greater detail, consider the replacement likelihood of the front-most worker. A marginal decline in $x_1$ results in two effects: (i) a direct cost saving of $c \frac{\zeta_1^x}{p_3-\zeta_1^x}$, and (ii) a direct incentive cost of $c \frac{p_3 (\zeta_1^x-p_0)}{(p_3-\zeta_1^x)^2}$. Under full utilization, the front-most worker knows that, if he is not replaced, the end-most worker will be replaced with certainty, implying a shirking-success rate $\zeta_1^x = p_1$. Plugging $\zeta_1^x = p_1$ into the two expressions and comparing them, we find that the direct incentive cost dominates the direct cost saving precisely when $p_1^2 > p_0 p_3$. A marginal reduction in $x_1$ therefore benefits the principal in this case. This implies that the principal prefers not to utilize the full AI capacity. 

An important implication of whether the principal fully utilizes or underutilizes AI capacity is how utilization shapes workers' beliefs about the presence and position of the AI agent in the production sequence. Moving from full to underutilization alters the uncertainty workers face, which can be used as a strategic tool that benefits the principal. Under full utilization, only the middle worker faces uncertainty about which worker is replaced: the front-most (end-most) worker knows that the end-most (front-most) worker must be the one replaced if he himself is not. Underutilization introduces an additional layer of uncertainty---all workers are now unsure whether a replacement has taken place at all.

To see how this uncertainty may benefit the principal, consider moving from full to underutilization by reducing $x_1$, while keeping $x_2=0$ and $x_3$ constant. Reducing $x_1$ lowers $w_1$ while leaving $w_2$ and $w_3$ unchanged, which can benefit the principal. At the same time, since $x_1$ is reduced, the principal needs to pay the front-most worker more frequently. When the former effect outweighs the latter, underutilization---and the resulting environment of uncertainty---can make the principal better off. A similar strategic use of uncertainty by a principal in organizations has been the focus of several recent theoretical studies \citep[e.g.,][]{halac2021rank, halac2025contracting}.

\paragraph{Impact of AI adoption on wages by worker's team position.} Having shown that the optimal AI adoption strategy in a team setting involves randomization, we next investigate how the wages and expected payoffs of workers change as the principal adopts AI, conditional on their position in the production sequence. Let $w^0_i$ denote the optimal compensation of worker $i$ in the absence of AI adoption, and $w^{x^*}_i$ the optimal compensation following optimal AI replacement. The following proposition shows that, while the wages of some workers may increase, the intra-team wage hierarchy of all-human teams is preserved under optimal AI adoption.

\medskip
\begin{proposition}
\label{prop: monotonic_wage}
Optimal AI adoption preserves the order of the wages based on the position of the workers ($w^{x^*}_3>w^{x^*}_2>w^{x^*}_1$). Moreover, after optimal AI adoption, worker 1's and worker 2's wages increase ($w^{x^*}_1>w^0_1$, $w^{x^*}_2>w^0_2$), and worker 3's wage remains unchanged ($w^{x^*}_3=w^0_3$).
\end{proposition}
\medskip

Proposition~\ref{prop: monotonic_wage} provides two key insights. First, optimal AI adoption increases the wages of the front-most and middle workers but leaves the wage of the end-most worker unchanged. The wage increase for the first two workers stems from the positive replacement probability of their successors, which makes shirking more appealing for them. Since the end-most worker has no successor, the success rate following his shirking does not depend on the AI adoption strategy, and his wage remains unchanged.

The second insight concerns the intra-team wage hierarchy. Proposition~\ref{prop: compensation} and the findings of \cite{winter2010transparency} imply that with all-human teams, optimal wages increase monotonically in the worker's position in the production sequence. A suboptimal adoption of AI can, in principle, disrupt this hierarchy---as we formally establish in Appendix~\ref{sec:Appendix_Supporting_WageHierarchy} (Lemma~\ref{lemma:AlteringWageHierarchy}). Proposition~\ref{prop: monotonic_wage} shows that under the optimal AI strategy, the wage hierarchy of all-human teams is preserved despite the upward pressure on wages.

Next, we investigate how the intra-team wage gap---the dispersion between the highest and lowest wage---changes following AI adoption. For a given replacement strategy $x$, we denote this wage gap by $\text{Gap}^x$:
$$\text{Gap}^x = \max\limits_{i} \{w_i^x\} - \min\limits_{i} \{w_i^x\}.$$
In a team with three workers, $\max_i \{w_i^x\} = w^{x^*}_3$ and $\min_i \{w_i^x\} = w^{x^*}_1$, so $\text{Gap}^{x^*} = w^{x^*}_3 - w^{x^*}_1$. A comparison of the wages indicates that the intra-team wage gap declines following optimal AI adoption, as Corollary~\ref{coro: wage_inequality} states formally.

\medskip
\begin{corollary}
\label{coro: wage_inequality}
The intra-team wage gap decreases following optimal AI adoption: $\text{Gap}^{x^*} < \text{Gap}^0$. 
\end{corollary} 
\medskip

Our finding that highlights the decline in the intra-team wage gap following technology adoption can be contrasted with earlier studies pointing to a widening in pay inequality from technology adoption, whether driven by workers' skill levels \citep{acemoglu2024polarization} or by the job a worker occupies within a firm \citep{barth2020robots}. The apparent contradiction is reconciled by noting that the workers in these studies differ substantially in their exposure to technology, while in our setting workers are identical apart from their position in the production sequence and make equal contributions to the firm's success. By abstracting from these confounds, we isolate the declining wage gap as a direct consequence of how AI adoption alters incentives within a team.

The findings thus far paint a somewhat optimistic picture of the effects of new technologies on organizations, with only partial job losses and wages weakly increasing. It is important, however, to also consider the effects on workers' overall earnings, particularly if workers face income losses during their ``off'' time. To this end, Proposition~\ref{prop: monotonic_paoff} summarizes the change in workers' payoffs---defined as $(1-x_i)(p_3 w_i^{x} - c)$ for worker $i$---comparing them before and after the optimal deployment of AI. 

\medskip
\begin{proposition}
\label{prop: monotonic_paoff}
Following the optimal AI strategy, the middle worker's payoff increases, the end-most worker's payoff decreases, and the front-most worker's payoff may increase or decrease relative to the payoffs in the absence of AI. 
\end{proposition}
\medskip

Given the probabilistic nature of work following optimal AI adoption, it is useful to consider expected payoffs alongside the wages earned when workers are not displaced. Proposition~\ref{prop: monotonic_paoff} completes this picture and underscores that even in homogeneous teams, AI adoption creates winners and losers: depending on their position, some workers are better off and some are worse off. Recall that the front-most and middle workers' wages increase, while the end-most worker's wage is unchanged. The front-most and end-most workers face partial replacement, while the middle worker does not. Combining these effects, the end-most worker---the highest earner---loses payoff, while the middle worker, who is essential for maintaining the information flow in the network and is therefore not replaced, sees his payoff rise. The front-most worker's payoff may increase or decrease depending on his replacement probability and wage increase. As a result, payoffs may no longer follow the same hierarchy as wages; we discuss these changes in more detail in the context of a specific production function in \ref{sec: A_complementary}.    

%% file: Extensions.tex
\section{Extensions} \label{sec: extensions}

\subsection{Team Synergies and Task Complementarity} \label{sec: complementarity}
A key aspect of teamwork is that the tasks carried out by individual team members build on each other to determine the success of the project. Our main model assumes such complementarity, but how does the degree of task complementarity shape the likelihood of AI replacement and workers' wages? To explore these questions, we adopt a production function that incorporates a parameter capturing the degree of task complementarity, the {\em O-ring} production function, which is well established in the literature \citep{kremer1993ring, winter2004incentives}. Let $p_k = \alpha^{n-k}$, where $\alpha \in (0,1)$ calibrates the degree of complementarity.\footnote{One may interpret this functional form as a task structure in which each worker is responsible for a task that succeeds with probability $1$ if he exerts effort and with probability $\alpha$ if he shirks. The project succeeds only if all tasks succeed.} The smaller $\alpha$ is, the more complementary the efforts are. Put differently, $\alpha$ indicates how consequential a single worker's effort is to project success. Proposition~\ref{closedform} provides the optimal replacement strategy and wages under this production function.

\medskip
\begin{proposition} \label{closedform}
When $p_k=\alpha^{n-k}$, the optimal replacement strategy is
\[
x_1^*=\frac{\sqrt{1+\alpha}-1}{\alpha}, \quad
x_2^*=0, \quad
x_3^*=1-x_1^*.
\]
As the degree of complementarity increases ($\alpha$ becomes smaller), the likelihood of the front-most (end-most) worker being replaced increases (decreases). Moreover, the workers' compensations satisfy
\[
w_1^{x^*} = \frac{c}{1-\alpha^2}, \quad
w_2^{x^*} = \frac{c}{(1-\alpha)\sqrt{1+\alpha}}, \quad
w_3^{x^*} = \frac{c}{1-\alpha}.
\]
\end{proposition}
\medskip

In addition to characterizing the optimal replacement strategy under the O-ring production function, Proposition~\ref{closedform} confirms our earlier general results: (i) the optimal strategy involves randomization (Proposition~\ref{prop: randomization}), and (ii) the middle worker is never replaced (Proposition~\ref{prop: x2=0}). Under this functional form, the principal chooses to fully utilize AI capacity, because the underutilization condition in Proposition~\ref{prop: useup} fails: $p_1^2 = \alpha^4$ and $p_0 p_3 = \alpha^3$, so $p_1^2 \leq p_0 p_3$.

The proposition indicates that as worker efforts become more complementary (smaller $\alpha$), the likelihood of replacing the front-most worker increases and that of the end-most worker decreases. In the extreme case when $\alpha \to 0$, both $x_1^*$ and $x_3^*$ converge to $0.5$. This is intuitive: when team efforts are highly complementary, a single worker's shirking is enough to compromise the success of the project, so the domino effect discussed in Section~\ref{sec: analysis} plays only a limited role in disciplining workers. When the domino effect is less potent, the wages of workers vary less by position, and the direct cost savings are similar whether the principal replaces the front-most or the end-most worker. The direct and indirect incentive costs become negligible for the same reason, so replacing either worker yields similar outcomes for the principal. Higher task complementarity therefore results in more homogeneous replacement, moving both the front-most and end-most workers' replacement probabilities toward $0.5$.

Turning to wages, as task complementarity increases (smaller $\alpha$), the wages of all workers decline. The intuition is straightforward: with greater complementarity, workers are intrinsically more motivated to exert effort, since shirking is more consequential, reducing the need for high-powered wage incentives.

Recall from Corollary~\ref{coro: wage_inequality} that optimal AI adoption reduces the intra-team wage gap. A natural follow-up question is how this inequality-reducing effect varies between teams with high- and low-complementarity tasks. Under the O-ring production function, in the absence of AI adoption, the wage gap is $\text{Gap}^{0} = \frac{c\alpha (1+\alpha)}{1-\alpha^3}$, and under the optimal AI strategy it is $\text{Gap}^{x^*} = \frac{\alpha c}{1-\alpha^2}$. Their ratio is:
$$
\frac{\text{Gap}^{x^*}}{\text{Gap}^{0}} = 1 - \frac{\alpha}{(1+\alpha)^2}.
$$
This ratio is always less than 1, confirming the reduction in intra-team wage inequality from AI adoption. At the same time, as task complementarity increases ($\alpha \to 0$), the ratio approaches 1, meaning that AI's potency in reducing the wage gap diminishes. Put differently, AI makes a bigger dent in wage inequality in teams where tasks are relatively independent and a single worker's shirking is less consequential for project success.

In \ref{sec: A_complementary}, we extend the analysis to compare payoffs across workers, in addition to comparing each worker's payoff before and after AI adoption. The analysis serves as a robustness check and shows that payoffs change in the direction discussed in Proposition~\ref{prop: monotonic_paoff}.

\subsection{Task-Based AI Substitution}
\label{sec:TaskBased}
Our analysis so far has focused on a scheme in which a worker is replaced by AI for his entire task. In real-world applications, however, AI often complements human labor by taking over only a subset of tasks, resulting in substitution at the task level rather than complete worker replacement. To capture this possibility, we extend the model to allow AI to substitute for a fraction of a worker's tasks. We show that, in equilibrium, the principal still prefers full worker-level replacement over partial task-level replacement.

Let $x_i$ denote the fraction of worker $i$'s tasks performed by AI, which proportionally reduces his effort cost to $(1-x_i)c$. For consistency with our main analysis, we assume an AI capacity of $1$:
\[
\bar{x} = \sum_{i=1}^n x_i \leq 1.
\]
Under this modification, each position in the network is occupied by a worker-AI pair. As in the main model, the success of the project depends on the number of positions \emph{effectively} exerting effort. If worker $i$ exerts effort, the entire pair is considered to exert effort. If the worker shirks, only the AI---which performs a fraction $x_i$ of the tasks---remains functional, so the pair contributes to effective effort with probability $x_i$. When $k$ worker-AI pairs effectively exert effort, the project succeeds with probability $p_k$. This adjustment allows the success probability to reflect partial AI substitution within each position.

We also modify the peer monitoring structure. The signal that the successor receives reflects the contribution of the worker-AI pair, rather than the worker alone. Because AI handles a fraction $x_i$ of worker $i$'s tasks, the successor receives a signal of shirking only when the worker-AI pair appears to be shirking, which occurs with probability $1-x_i$. Partial AI substitution thus partially conceals the worker's shirking, lowering the probability of detection.

Finally, we reconfigure the compensation structure. Under task-based AI adoption, compensation is no longer tied to a worker's entire role but is instead determined per unit of task performed. To make this distinction clear, we now use $w_i$ to denote the compensation per task carried out by worker $i$. If a fraction $x_i$ of his tasks is handled by AI, his total compensation upon project success is $(1 - x_i) w_i$. For any pair $(w, x)$ that induces joint effort among workers, the principal's expected compensation cost remains as in Equation~\eqref{Wxw}.

\paragraph{Worker incentives and optimal compensation.} The following result, a counterpart to Proposition~\ref{prop: compensation}, characterizes the optimal compensation scheme for a given task-based AI adoption strategy $x$.

\begin{proposition} \label{prop:TaskBasedCompensation}
Fix a task-based AI adoption strategy $x = (x_1, \ldots, x_n)$. The optimal compensation scheme conditional on $x$, denoted by $w^x = (w_1^x, \ldots, w_n^x)$, satisfies
\[
w_i^x = \frac{c}{p_n-\zeta_i^x}, \quad \text{for each $i \in \{1, \ldots, n\}$ with $x_i<1$,}
\]
where
\begin{align*}
\zeta_i^x=x_i p_n+\sum_{k=1}^{n-i} \left[ p_{n-k} \, x_{i+k}\prod_{j=i}^{i+k-1}\left(1-x_j\right) \right]  + p_{i-1}\prod_{j=i}^n\left(1-x_j\right).
\end{align*}
\end{proposition}
\medskip

Proposition~\ref{prop:TaskBasedCompensation} characterizes the optimal compensation for workers whose tasks are not entirely performed by AI, as otherwise compensation is irrelevant. As in the main model, the optimal compensation scheme induces effort by supporting a trigger strategy profile as an equilibrium and ensures that workers remain indifferent on the equilibrium path. The compensation of worker $i$ decreases in $p_n - \zeta_i^x$, the difference in the probability of success when he exerts effort versus when he shirks. The explicit value of this difference is
\begin{align*}
p_n - \zeta_i^x &= (1-x_i) p_n -\sum_{k=1}^{n-i} \left[ p_{n-k} \, x_{i+k}\prod_{j=i}^{i+k-1}\left(1-x_j\right) \right]  -p_{i-1}\prod_{j=i}^n\left(1-x_j\right) \\
&= (1-x_i) \left(p_n -\sum_{k=1}^{n-i} \left[ p_{n-k} \, x_{i+k}\prod_{j=i+1}^{i+k-1}\left(1-x_j\right) \right]  -p_{i-1}\prod_{j=i+1}^n\left(1-x_j\right) \right).
\end{align*}
The key insight is that worker $i$'s per-task compensation decreases proportionally with $1 - x_i$, so his total compensation, $(1 - x_i) w_i^x$, remains unchanged regardless of how many tasks he retains. This is because AI replacement affects production and monitoring in parallel: while it lowers the worker's effort cost by offloading tasks, it also reduces the detectability of shirking. As a result, unless a worker's entire task is replaced, partial replacement---where AI handles only a fraction of his tasks---offers no clear advantage to the principal. This intuition motivates the formal analysis that follows.

\paragraph{Optimal task-based AI substitution.} For conciseness, we provide expressions for the wages in a three-worker sequential production team---the front-most, middle, and end-most workers, respectively:
\begin{align*}
w^x_1 &= \frac{c}{(1-x_1)\left(p_3 - \left[ x_2 p_2 + (1-x_2)x_3 p_1 + (1-x_2)(1-x_3)p_0\right] \right)}, \\
w^x_2 &= \frac{c}{(1-x_2)\left(p_3 - \left[ x_3 p_2 + (1-x_3) p_1 \right] \right)}, \\
w^x_3 &= \frac{c}{(1-x_3) \left(p_3 - p_2 \right)}.
\end{align*}
Restricting attention to replacement strategies in which no worker is fully replaced ($x_i < 1$ for each $i$), the principal's problem of minimizing expected compensation cost is
\begin{align*}
\min_{\substack{x_1,x_2,x_3 \in [0,1) \\ x_1+x_2+x_3 \leq 1}} \quad (x_1+x_2+x_3)\, c &+ \frac{p_3 c}{p_3- \left[ x_2 p_2 + (1-x_2)x_3 p_1 + (1-x_2)(1-x_3)p_0\right]} \\
&+ \frac{p_3 c}{p_3- \left[ x_3 p_2 + (1-x_3) p_1 \right]} + \frac{p_3 c}{p_3- p_2}.
\end{align*}

The solution to this restricted problem yields $(x_1, x_2, x_3) = (0,0,0)$, since the partial derivatives with respect to $x_1$, $x_2$, and $x_3$ are all strictly positive. To determine the globally optimal strategy, we compare this solution with the strategies that involve full replacement of a single worker: $(x_1, x_2, x_3) \in \{ (1,0,0), \, (0,1,0), \, (0,0,1)\}$. The optimal solution in the task-based setting thus reduces to selecting the best full-replacement strategy, which coincides with the optimal pure-replacement strategy from the main model. As established earlier, this involves fully replacing either the front-most or the end-most worker as long as $p_0>0$. If $p_0=0$, then not replacing anyone, $x=(0,0,0)$, is also optimal: the front-most worker can be induced to exert effort with no information rent (his expected payment equals the effort cost $c$), because if he shirks, all subsequent workers also shirk and the success probability is $p_0=0$. Proposition~\ref{prop:TaskBasedOptimal} formalizes this result.

\medskip
\begin{proposition}\label{prop:TaskBasedOptimal}
The optimal task-based AI adoption strategy involves fully replacing either the entire task of the front-most or the entire task of the end-most worker, when $p_0>0$. If instead $p_0=0$, then replacing no one, $(0,0,\dots,0)$, is also optimal.
\end{proposition}
\medskip

Although this alternative model allows AI to handle a fraction of a worker's tasks, partial replacement does not materialize in equilibrium: the principal prefers full replacement of workers' tasks. This outcome aligns with what is commonly observed in practice, where workers' tasks tend to be displaced in their entirety, which is equivalent to replacing the worker himself.

Finally, we turn to the utilization of AI capacity and ask what differs in task-based replacement relative to the benchmark model. Consider a scenario in which the principal's total AI capacity is $\mathcal{A} \in (0, n]$, and the replacement strategy $x = (x_1, \ldots, x_n)$ satisfies $\sum_{i=1}^n x_i \leq \mathcal{A}$. As in the main model, fractional task-based replacement weakens peer monitoring and raises compensation costs. Consequently, the principal prefers to replace the entire task of some workers while letting the others carry out their tasks in full, rather than allocating AI fractionally. Corollary~\ref{cor:nonunit_capacity} provides the solution for a general AI capacity, following the same approach as in Proposition~\ref{prop:TaskBasedOptimal}.

\medskip
\begin{corollary} \label{cor:nonunit_capacity}
The optimal AI adoption strategy allows for no fractional task allocation and replaces all tasks of $\lfloor\mathcal{A} \rfloor$ workers. The replacement is done such that the remaining workers, whose tasks are not replaced, are consecutive to each other. If $p_0=0$ and $\lfloor\mathcal{A} \rfloor =1$, then it is also optimal to not replace anyone, $(0,0,\ldots,0)$.
\end{corollary}
\medskip

An immediate implication of the corollary is the possibility of underutilization of AI capacity. For example, if $\mathcal{A} < 1$, the optimal strategy is to refrain from adopting AI altogether. This finding reinforces the point that even under a limited AI capacity, the principal may, under some circumstances, choose not to exhaust it.

\subsection{Network Structure} \label{sec: network}

In the main part of the paper, we focus on an organization with a chain network structure, whose optimality is established in \cite{winter2010transparency}. The chain structure well represents organizations where tasks are carried out sequentially among team members. Naturally, for reasons we abstract from in this paper, teams may be organized in other structures. Another commonly observed structure is the {\em star} network, in which several team members carry out their tasks simultaneously, followed by a central---possibly higher-level---executive who carries out his own. In this section, we extend our framework to the star network and summarize the insights that are identical to and different from the chain environment.

Consider a team in which $n-1$ {\em peripheral} workers (indexed $1, \ldots, n-1$) each generate a signal about their contribution that is received by a {\em central} worker---possibly a manager---indexed $n$. Proposition~\ref{prop: star_optimum} demonstrates the robustness of two of our earlier insights: (i) a team member who is essential to maintaining the information flow between team members is less likely to be replaced, and (ii) a randomized replacement strategy may still be optimal.

\medskip
\begin{proposition} \label{prop: star_optimum}
In a star structure, all AI replacement strategies satisfying $\sum_{i=1}^{n-1}x_i=1$ and $x_n=0$ are optimal, if the condition $p_{n-1}^2-p_n p_{n-2}\leq 0$ is satisfied.
\end{proposition}
\medskip

The proposition shows that the central worker---the team member essential for maintaining the information flow---is never replaced, while the peripheral workers face a higher risk of being replaced with AI. Since the peripheral workers are identical, any AI replacement strategy in which their replacement probabilities sum to the AI capacity is optimal, provided the sufficient condition $p_{n-1}^2 - p_n p_{n-2} \leq 0$ holds. The key insights from our benchmark model, while shaped by the characteristics of the network, are thus not specific to the chain structure and carry forward to the star organization as well.

\subsection{Strategic AI} \label{sec: strategicAI}

So far, we have considered AI agents that, when deployed, automatically exert effort. We now extend the analysis to a setting in which AI's behavior can be programmed to condition its action on the actions of its predecessor, much as human workers do. Rather than exerting effort unconditionally, a programmable AI can follow a strategy that maps its predecessor's action into its own effort decision. This extension allows us to study the implications of an AI that replicates not only the productivity but also the strategic responsiveness of human agents within the sequential chain structure.

The strategy of an AI agent placed at position $i \in \{2, \ldots, n\}$, $\sigma_i^{AI}: \{0,1\} \mapsto \{0,1\}$, specifies its effort decision based on the signal it receives about its predecessor's action $e_{i-1}$. Since there is no predecessor at position $1$, the AI strategy in this position is $\sigma_1^{AI} \in \{0,1\}$. The principal's AI deployment decision now has two dimensions: (i) which human worker to replace, and (ii) how to program the AI to act strategically, depending on its position in the sequence. We solve for the principal's decision under the assumption that she can costlessly program AI.

Recall that the principal's objective is to induce all workers to exert effort while keeping compensation costs as low as possible. For each worker $i$, the required compensation depends on $\zeta_i$, the project's success probability from worker $i$'s perspective if he chooses to shirk. Minimizing compensation therefore requires minimizing $\zeta_i$, which is achieved by programming the AI to shirk whenever its signal indicates that its predecessor shirked, regardless of its position in the sequence. Hence, the principal's optimal programming of AI actions is
\[
\sigma_1^{*AI} = 1, \quad 
\sigma_i^{*AI}(e_{i-1}) = 
\begin{cases}
1 & \text{if } e_{i-1} = 1, \\ 
0 & \text{if } e_{i-1} = 0,
\end{cases}
\quad \text{for all } i \in \{2, \ldots, n\}.
\]
Given this optimal programming, the principal then determines which agent to replace. We provide the optimal replacement strategy in Proposition~\ref{prop: StrategicAI}.

\medskip
\begin{proposition} \label{prop: StrategicAI}
When the principal can strategically program AI to condition its effort on the actions of its predecessor, the optimal deployment of AI is $x=(0,0, \ldots, 1)$. That is, the principal replaces the end-most worker with certainty.
\end{proposition}
\medskip

With a strategic AI, the principal no longer relies on randomization and strictly chooses to replace the end-most worker. She does so because, in this environment, she is no longer subject to the direct or indirect incentive costs and is concerned only with direct cost savings, which leads her to replace the agent with the highest compensation. This result follows immediately in a world where AI can perfectly and costlessly mimic human strategic behavior; given these advantages, the principal will always use AI. In practice, however, programming AI to mimic humans is unlikely to be costless. Doing so amounts to endowing AI with both monitoring and decision-making capacities, which requires upfront and continuous investments. If the costs of programming AI were prohibitively high, we would revert to the benchmark setting and the benchmark deployment strategy.

These findings yield two key managerial insights. First, the optimal deployment of AI in human-AI teams requires attention to the strategic programming of AI---specifically, how it should respond to shirking by humans (or other AI agents). To our knowledge, such monitoring and adaptive response mechanisms are not yet embedded in existing AI systems. Second, and perhaps counterintuitively, a principal may sometimes prefer an AI that shirks occasionally---and is therefore less efficient---over one that always works unconditionally. For managers, this suggests that evaluating the benefits of automation requires accounting for such strategic shirking, since it reduces the total productive capacity of these technologies.

\subsection{Additional Discussions} \label{sec: AdditionalDiscussions}

\paragraph{Unitary capacity of AI.} \label{sec:UnitaryCapacity} Our main model assumes that the principal deploys a single AI agent within the organization. Suppose instead that she can deploy $\mathcal{A}$ AI agents, for some integer $\mathcal{A}$. The question is how these agents should be deployed and whether randomization remains optimal.

First, note that the optimal pure deployment strategy with $\mathcal{A}$ AI agents replaces $\mathcal{A}$ human workers while preserving the connection among the remaining human workers. This can be achieved by placing the $\mathcal{A}$ AI agents at the two ends of the production line (some at the beginning and some at the end), ensuring that no AI agent is placed in the middle, which would break the continuity of the human segment.

We now examine whether randomization can improve upon this optimal pure strategy. Consider the following two-step procedure that explicitly incorporates randomization:
\begin{enumerate}
    \item \textit{Construct a contiguous human segment.} Deploy $\mathcal{A}-1$ AI agents to replace workers such that the remaining human workers form a single connected segment of length $n-(\mathcal{A}-1)$.
    \item \textit{Randomized deployment of the final AI agent.} Replace the first worker in the connected human segment with probability $\rho$, and the last worker with probability $1-\rho$.
\end{enumerate}
From the baseline analysis, this mixed strategy strictly dominates the optimal pure replacement strategy.\footnote{This holds provided that the human segment constructed in step 1 contains at least three workers, i.e., $n-(\mathcal{A}-1) \geq 3$.} The optimality of randomized replacement therefore extends naturally to the case of multiple AI agents.

\paragraph{Limited capacity of AI.} \label{sec:LimitedCapacity} In the main model, we assume AI capacity is limited. This realistic assumption captures the resource constraints of any organization. Adoption and use of AI technologies are both costly, so while organizations are trying to integrate AI into multiple aspects of their work, they prioritize the use of AI for the tasks where the returns to the system are highest \citep{Stackpole2024BusinessAIOpportunities}.

In our main model, we purposely abstract away from the costs of AI adoption and use in order to focus on the deployment strategy itself and sharpen the question of how to prioritize among nearly identical workers when deploying AI.

These said, in some distant future costs of adopting and using AI may no longer be a significant constraint to achieve AI capacity, and a principal may deploy AI in a way to replace all workers. We also introduced a version of this scenario when we investigated task-based replacement in Section \ref{sec:TaskBased} by relaxing the capacity to take an arbitrary value $\bar{x} \leq n$.

If we follow our benchmark model, in the absence of a capacity constraint and AI adoption/use costs, trivially, the principal would choose to replace all workers. In this setting, all other insights we deliver---from wages, pay inequality, randomization of work---are no longer relevant.

\paragraph{Imperfect observability of AI utilization.} \label{sec:ImperfectObservability} In the baseline model, workers do not observe the realized outcome of the AI replacement strategy: the signal a worker receives about his predecessor's contribution is silent on whether that contribution came from a human or from an AI. This scenario is analogous to an AI system successfully passing a Turing test, where the output of the preceding task cannot be reliably attributed to either source. Recent evidence indicates that AI can today generate outputs that are ``statistically indistinguishable from a random human'' \citep[][abstract]{mei2024turing}. In many instances of AI integration into organizational workflows, workers may therefore fail to identify whether the task-relevant inputs they receive are from human colleagues or from AI systems. For example, call center agents increasingly rely on written guidance generated by AI \citep{li2024generative}, yet they often cannot determine whether the instructions were produced by a more experienced employee or by an algorithm. Likewise, software developers frequently collaborate with AI systems \citep{HoffmannGenAIWork}, receiving code suggestions and feedback that are not readily distinguishable from those produced by human peers. More behavioral evidence for the indistinguishability of human vs. AI teammate output comes from a controlled experiment by \cite{yan2025social}. The authors report that participants working in groups correctly identified which of their fellow group members were AI agents only 30.8$\%$ of the time, while misclassifying human teammates as AI in 27.6$\%$ of cases. These findings demonstrate that in team interactions, workers lack the ability to perfectly distinguish human collaborators from AI peers---a condition that aligns closely with the setting analyzed in this paper.

While both the literature and practice highlight that it is nearly impossible to distinguish a human input from an AI input, the setting we study can also arise {\em optimally}, motivated by a micro-model. If the principal could choose the work setting, she would optimally design one with imperfect observability. Thus, imperfect observability is a {\em feature} of the environment that arises optimally, rather than a bug.

To see this, consider an alternative timeline that facilitates observability: the principal first chooses a (potentially mixed) AI replacement strategy; the realized replacement is then publicly observed; the principal next chooses a compensation scheme; finally, workers choose their efforts and the outcome of the project is realized.

In this setting, the principal's payoff associated with any replacement strategy $x$ is linear in pure strategies. Moreover, as we have shown, there are two optimal pure strategies when $p_0>0$: replacing the front-most and replacing the end-most worker. Any randomization between these two is therefore also optimal. When $p_0=0$, not replacing any worker is equally optimal, so randomization between the two replacement strategies and no replacement is optimal as well. Since payoffs from randomization equal those from pure replacement, the principal has no strict preference for randomization. These observations together suggest that imperfect observability benefits the principal.

Another implication is that, under perfect observability, strategic uncertainty---shaping workers' beliefs about replacement---is no longer feasible. Underutilization of AI resources is therefore no longer optimal, and the principal fully utilizes the AI capacity as long as $p_0>0$. When $p_0=0$, not replacing any worker is also optimal, alongside fully utilizing AI capacity to replace the front-most or end-most worker, so underutilization of AI capacity may arise in equilibrium.

\paragraph{Heterogeneity of team members.} \label{sec: heterogeneity} In the main model, we purposefully abstract from creating asymmetries among workers other than their position in the network, and we assume that each task contributes identically to the project's success. We do so to demonstrate that informational position alone is sufficient to generate differential risks of AI replacement and wage outcomes. Team members may, however, vary in other dimensions, for instance in their skills or the cost of effort, and the contributions of tasks to project success may vary by position. Without building a full-fledged model, we can argue that the replacement probabilities suggested by our propositions may look different in such settings; in particular, the middle worker may face a non-zero probability of replacement if he is more costly or less skilled than his peers.

%% file: Conclusion.tex
\section{Conclusion} \label{sec: conclusion}

Owing to the rapid growth of robotization and AI, organizations are going through a profound transformation about what work is and how work is done. New AI technologies have transformed how employees interact with each other as well, as their tasks are partially or fully carried out by the AI agents. Given the scale and scope of this transformation, it is essential to ask how work and groups of employees should be re-optimized around these new technologies.

There is a growing literature on how machines alter worker and business performance \citep{otis2024uneven} and human decision-making \citep{de2023your, boyaci2024human, LLLZ2024, Burtch2025police}. Much of the existing literature focuses on individual-level performance effects of the workers who adopt AI \citep[e.g.,][]{dell2023super} or the labor displacement and wage changes due to AI \citep[e.g.,][]{chen2025displacement, humlum2025large}. 
We contribute to these understudied areas by focusing on the effects of AI agents on worker productivity from a team perspective, explicitly modeling how adoption of AI alters the motivation of the workers and the incentives set by the principal.
We demonstrate how a principal integrates AI in a team previously optimized for human employees, and characterize the replacement risk and resulting wage changes for each team in the production sequence based on their position.

First, our analysis reveals that optimal AI adoption involves stochastic rather than deterministic labor force replacement.  AI resources are best utilized to complement the workers' roles, possibly in a way that alters or reduces the workload, while maintaining the critical benefits of human worker input. 
This finding aligns with recent theoretical work emphasizing the importance of human-AI complementarity in maximizing organizational productivity \citep{brynjolfsson2017can}.   

Second, we demonstrate that the likelihood of AI replacement varies by a worker's position within the team's information network. To maintain essential information flows and leverage peer monitoring effectively, it is optimal to minimize the replacement risk for the middle worker, while setting positive replacement probabilities for the front-most and end-most team members. The organizational position relates to the recent literature emphasizing the critical role of information and communication in teams for organizational efficiency \citep{garicano2000hierarchies, angelucci2017motivating, matouschek2025organizing}. 

Third, we investigate the wage implications of optimal AI integration within production teams. In his analysis of the AI effects on wages and pay inequality, \cite{acemoglu2025simple} states that while theoretically it is possible for AI to reduce wages and exacerbate inequality, there is little evidence in practice for these predictions. Our findings provide additional theoretical predictions relevant to this domain; by suggesting that the strategic adoption of AI does not necessitate wage reductions; instead, managers should anticipate comparable or higher wages post AI adoption relative to the wages before AI adoption. Importantly, while existing wage hierarchies remain unchanged following AI adoption, the intra-team wage gap is significantly reduced as lower-paid workers receive wage increases. Collectively, these insights challenge prevailing assumptions and provide practical managerial guidance for navigating AI integration. Table~\ref{tab:insights} summarizes these key managerial prescriptions.

\medskip
\begin{table}[t!]
    \centering
        \caption{Summary of Key Managerial Insights}
        \begin{adjustbox}{width=\textwidth}
    \begin{tabular}{l|l}
\toprule
Question  & Findings from our setting  \\
\hline
How should a manager optimally integrate AI  & AI deployment to replace workers should be  randomized, \\
to a human team?  & e.g., across different projects or work shifts. \\

& \\
Which workers in the team are more & The front-most and the end-most workers face \\
likely to be replaced with AI?  &  higher replacement risk.  \\
&  The middle worker should not be replaced. \\
  & \\

How does AI adoption affect worker wages? & Wages increase for all except the end-most worker,  \\   
& whose wage remains unchanged.  \\
& \\

How does the optimal deployment of AI & Intra-team wage gap  declines following \\   
impact intra-team wage inequality?  &  optimal AI adoption. \\
& \\

When is AI's wage inequality reducing effect  & AI's wage inequality reducing effect is more potent  \\   
more potent? & in teams with lower complementarity.  \\
& \\

Is it better for AI to carry out a fraction of a & No,  the principal prefers to replace a worker  \\   
worker's tasks, rather than fully substituting him? &  rather than a fraction of tasks. \\
          \bottomrule
    \end{tabular}
\end{adjustbox}

    \label{tab:insights}
\end{table}


While AI stands to offer productivity gains, our findings highlight that the integration of AI into human systems must be carefully designed to maintain the information flows between the team members. By doing so, organizations can benefit from replacing workers with AI while minimizing the negative externalities on the non-replaced members of the team.

\bigskip

%% file: Appendix.tex
\clearpage
\newpage
\appendix

\setstretch{1.3}

\setcounter{page}{1}
\setcounter{table}{0}
\setcounter{figure}{0}
\pagenumbering{arabic}
\renewcommand*{\thepage}{A\arabic{page}}

\renewcommand{\thesection}{Appendix \Alph{section}}
\renewcommand{\thesubsection}{\Alph{section}.\arabic{subsection}}

\renewcommand*{\thetable}{\thesection\arabic{table}}
\renewcommand*{\thefigure}{\thesection\arabic{figure}}

\counterwithin{figure}{section}
\counterwithin{table}{section}

\begin{center}
    \Large{\bf Appendix}
\end{center}

\section{Supporting Statements}
\label{sec:AppendixA}
This appendix section presents supporting statements omitted from the main text for brevity, along with their corresponding proofs.

\subsection{Existence}

We first establish the existence of the optimal AI replacement.

\begin{lemma}
The optimal replacement strategy exists.\label{lem:existence}
\end{lemma}

\begin{proof}[\textbf{Proof of Lemma~\ref{lem:existence}}]  \hfill

\noindent We  first prove that the principal's payoff is continuous with respect to $x$. If $x_i<1$ for all $i$, then
\[
\frac{W_x}{c}=\sum_{i=1}^n x_i+\sum_{i=1}^n \frac{\left(1-x_i\right) p_n}{ p_n-\zeta_i^x}.
\] Let $\vec{x}_j$ be a vector whose $j$th element is 1 and all other elements are zeros, then
\[
\lim_{x\rightarrow \vec{x}_j} \frac{W_x}{c}=1+\sum_{i\neq j}\frac{ p_n}{ p_n-\zeta_i^{\vec{x}_j}}+\lim_{x\rightarrow \vec{x}_j} \frac{\left(1-x_j\right) p_n}{ p_n-\zeta_j^x}.
\] 
As can be seen in equation~\eqref{eq: Ai}  when $x_j=1$, the zero denominator makes $\zeta_j^x$ not well-defined. \eqref{eq: Ai} implies
\[
p_{j-1}\leq \zeta_j^x\leq p_{n-1}
\] and thus
\[
\frac{\left(1-x_j\right) p_n}{ p_n-p_{j-1}}\leq \frac{\left(1-x_j\right) p_n}{ p_n-\zeta_j^x}\leq \frac{\left(1-x_j\right) p_n}{ p_n-p_{n-1}}.
\]  Note that
\begin{gather*}
\lim_{x\rightarrow \vec{x}_j}\frac{\left(1-x_j\right) p_n}{ p_n-p_{j-1}}=0\\
\lim_{x\rightarrow \vec{x}_j}\frac{\left(1-x_j\right) p_n}{ p_n-p_{n-1}}=0.
\end{gather*} By squeeze theorem, 
\[
\lim_{x\rightarrow \vec{x}_j} \frac{\left(1-x_j\right) p_n}{ p_n-\zeta_j^x}=0. 
\] So
\[
\lim_{x\rightarrow \vec{x}_j} \frac{W_x}{c}=1+\sum_{i\neq j}\frac{ p_n}{ p_n-\zeta_i^{\vec{x}_j}},
\] which is indeed the principal's payment when $x=\vec{x}_j$. 

Moreover, note that the set of  feasible replacement strategies, 
\[
\left\{\left(x_1, \cdots, x_n\right) \vert x_i\geq 0, \forall i=1, \cdots, n, \sum_{i=1}^n x_i \leq 1\right\},
\] is a compact set. Therefore minimum of a continuous function on a compact set exists.
\end{proof}

\subsection{The Optimal Pure Replacement Strategy}
The next statement characterizes the optimal pure replacement strategy.

\begin{lemma}
The optimal pure replacement strategies involve replacing either the front-most or the end-most worker. That is, $(1,0,\dots,0)$ and $(0,\dots,0,1)$ are optimal when $p_0>0$. If instead $p_0=0$, then replacing no one, $(0,0,\dots,0)$, is also optimal.\label{topbottom}
\end{lemma}

\begin{proof}[\textbf{Proof of Lemma~\ref{topbottom}}] \hfill

\noindent
With slight abuse of notation, denote the principal's expected compensation cost under the pure replacement strategy in which worker $i \in \{1,\cdots, n\}$ is replaced with probability 1 by $W_i$ and that under the pure replacement strategy in which no worker is replaced by $W_\emptyset$. We will show that:
\begin{align*}
W_1=W_n = min \{ W_\emptyset, W_1, W_2, \cdots, W_n\}.    
\end{align*}
We have:
\begin{align*}
\frac{W_\emptyset}{c}&=\sum_{i=1}^n \frac{ p_n}{ p_n-p_{i-1}}  \\
\frac{W_i}{c} &=1+\sum_{k=1}^{i-1} \frac{ p_n}{ p_n-p_{n+k-i}}+\sum_{k=i+1}^{n} \frac{ p_n}{ p_n-p_{k-1}}, \text{ for each } i \in \{1,\cdots, n\}. 
\end{align*}
Then:
\[
\frac{W_{i+1}-W_i}{c}=\frac{ p_n}{ p_n-p_{n-i}}-\frac{ p_n}{ p_n-p_i}, \text{ for each } i \in \{1,\cdots, n-1\}.
\]
Therefore:
\begin{equation*}
W_{i+1} \begin{cases}
> W_i, & \text{if } i < n/2, \\
= W_i, & \text{if } i = n/2, \\
< W_i, & \text{if } i > n/2,
\end{cases} \text{ for each } i \in \{1,\cdots, n-1\}.
\end{equation*}
Note that \begin{align*}
\frac{W_{n+1-i}}{c}&=1+\sum_{k=1}^{n+1-i-1} \frac{ p_n}{ p_n-p_{n+k-\left(n+1-i\right)}}+\sum_{k=n+1-i+1}^{n} \frac{ p_n}{ p_n-p_{k-1}}\\
&=1+\sum_{k=i+1}^{n} \frac{ p_n}{ p_n-p_{k-1}}+\sum_{k=1}^{i-1} \frac{ p_n}{ p_n-p_{n+k-i}}.\\
&=\frac{W_i}{c}
\end{align*} This implies that $W_1=W_n$. 
Moreover, when $p_0>0$ we have 
\[
\frac{W_1}{c}=1+\sum_{i=2}^{n} \frac{ p_n}{ p_n-p_{i-1}}<\frac{ p_n}{ p_n-p_0}+\sum_{i=2}^{n} \frac{ p_n}{ p_n-p_{i-1}}=\frac{W_\emptyset}{c}.
\]
and when $p_0=0$, we have:
\[
\frac{W_1}{c}=1+\sum_{i=2}^{n} \frac{ p_n}{ p_n-p_{i-1}}=\frac{ p_n}{ p_n-p_0}+\sum_{i=2}^{n} \frac{ p_n}{ p_n-p_{i-1}}=\frac{W_\emptyset}{c}.
\]
This concludes the proof.
\end{proof}

\subsection{Possibility of Altering Wage Hierarchy}
\label{sec:Appendix_Supporting_WageHierarchy}
The next result shows that (suboptimal) AI replacement can alter the wage hierarchy within the production team compared to the case where no replacement occurs. 
\begin{lemma}
For any $i<j<n$, there exists a replacement strategy $x$ such that $w^x_i>w^x_j$. \label{lemma:AlteringWageHierarchy}
\end{lemma}
\begin{proof}[\textbf{Proof of Lemma~\ref{lemma:AlteringWageHierarchy}}] \hfill 

\noindent Consider the replacement strategy with $x_i+x_{i+1}=1$ and $x_i \in \left(0,1\right)$. Then from equation~\eqref{eq: zeta} we have $\zeta^x_i=p_{n-1}$ and $\zeta^x_j=p_{j-1}$, which implies $w^x_i>w^x_j$.
\end{proof}

\bigskip 

\section{Proofs of Results}
This appendix section provides the proofs of the results stated in Section~\ref{sec: analysis} and Section~\ref{sec: extensions}.

\subsection{Proofs of Section~\ref{sec: analysis}}

\begin{proof}[\textbf{Proof of Proposition~\ref{prop: compensation}}]\hfill \\
\noindent The proof proceeds in three steps. We show the following in sequence:
\begin{enumerate}
\item Trigger strategy profile $\sigma^*$ constitutes an equilibrium under  $w^x$. 
\item  No alternative compensation scheme with a lower expected
cost can support $\sigma^*$ as an equilibrium.  
\item  Supporting any other strategy profile, apart
from the trigger strategy profile, that results in joint effort, as an equilibrium necessarily incurs a higher
expected compensation cost for the principal.
\end{enumerate}

\noindent \underline{Step 1.}
We will prove that, under the compensation scheme $w^x$, the trigger strategy is a best response for each worker $i \in \{1, \ldots, n\}$ given that all other workers follow the trigger strategy.

If worker $i$ observes a signal indicating that $e_{i-1} = 1$, he believes that nobody has deviated. In this case, choosing effort is a best response for worker $i$ if and only if: 
\begin{align}
p_n \, w_i - c \geq \zeta_i^x \,w_i . \label{eqn1} 
\end{align}
The left-hand side is the expected payoff for worker $i$ when he exerts effort, which leads all subsequent workers to exert effort, resulting in a $p_n$ probability of success. The right-hand side is his expected payoff when he shirks, where $\zeta_i^x$ is the probability of success from his perspective given that the other workers follow the trigger strategy profile. 

Note that $w^x_i$ is defined so that condition~\eqref{eqn1}, a necessary condition for the trigger strategy to comprise an equilibrium, holds with equality: $w_i^x = \frac{c}{p_n - \zeta_i}$ (equation~\eqref{eq: w}).  Therefore worker $i$ has no incentive to deviate from the trigger strategy after observing  $e_{i-1} = 1$.

If worker $i$ observes a signal indicating that $e_{i-1} = 0$, we are off the equilibrium path. Worker $i$ then believes that all the predecessors of worker $i-1$ exerted effort and worker $i-1$ is the first to shirk.\footnote{All subsequent arguments would analogously hold even if worker $i$ has different out-of-path beliefs regarding the behavior of his predecessors.} 
In this case, shirking ($e_i =0$) is a best response for worker $i$ if and only if:
\begin{equation}
p_{n-1} w_i^x -c \leq \hat{\zeta}_i^x \, w_i^x. \label{eqn2}
\end{equation}
The left-hand side is the expected payoff of worker $i$ when he exerts effort. By exerting effort, worker $i$ ensures that all his successors will also exert effort, resulting in a probability of success of $p_{n-1}$, as all workers other than $i-1$ have exerted effort. The right-hand side denotes the expected payoff for worker $i$ when he shirks, where $\hat{\zeta}_i^x$ is the probability of success when worker $i$ shirks, given that he has observed a signal indicating that $e_{i-1} = 0$ and believes that all workers preceding $i-1$ have exerted effort, and all his successors will adhere to trigger strategies. That is:
\begin{align*}
\hat{\zeta}_i^x &= p_{i-2}  \sum\limits_{k=1}^{i-2} \frac{x_k}{1 - x_i -x_{i-1}} \, + \, \sum\limits_{k=i+1}^n \frac{x_k}{1 - x_i-x_{i-1}} p_{n-1-k+i} \, + \, p_{i-2} \, \frac{1 - \bar{x}}{1-x_i-x_{i-1}}  \\
&= p_{i-2} \frac{1- \sum\limits_{k=i-1}^n x_k}{1-x_i -x_{i-1}} + \sum\limits_{k=i+1}^n \frac{x_k}{1 - x_i-x_{i-1}} p_{n-1-k+i} 
\end{align*}
More precisely, after observing $e_{i-1} = 0$, worker $i$ knows that neither he nor worker $i-1$ has been replaced. Therefore, the probability of worker $k$ being replaced (for $k \neq i, i-1$) is given by $\frac{x_k}{1 - x_i - x_{i-1}}$, while the probability of no replacement from worker $i$'s perspective is $\frac{1 - \bar{x}}{1 - x_i - x_{i-1}}$.
When the replaced worker is a predecessor of worker $i$ (and $i-1$) (i.e., when $k < i-1$), or when there is no AI replacement, all the successors of worker $i$ will shirk upon observing that worker $i$ shirks, resulting in a success probability of $p_{i-2}$. When the replaced worker is a successor of worker $i$ (i.e., when $k > i$), all the workers between $i$ and $k$ will shirk, and all the successors of $k$ will exert effort, resulting in a success probability of $p_{n-1-k+i}$.

We already know that $ w_i^x ( p_n- \zeta_i^x) = c$. Therefore showing that $ p_n- \zeta_i^x \geq p_{n-1} - \hat{\zeta}_i^x$ would suffice to show that condition \eqref{eqn2} holds, and hence each worker finds it optimal to shirk upon observing that his immediate predecessor shirks. 
Therefore, we would like to show that 
$$ p_n -p_{n-1} \geq \zeta_i^x- \hat{\zeta}_i^x$$
Note that 
\begin{align*}
 \zeta_i^x- \hat{\zeta}_i^x &= \underbrace{\left(p_{i-1} \frac{1- \sum\limits_{k=i}^n x_k}{1-x_i} + \sum\limits_{k=i+1}^n \frac{x_k}{1 - x_i}p_{n-k+i}\right)}_{\zeta_i^x}\\ & - \underbrace{\left(p_{i-2} \frac{1- \sum\limits_{k=i-1}^n x_k}{1-x_i -x_{i-1}} + \sum\limits_{k=i+1}^n \frac{x_k}{1 - x_i-x_{i-1}} p_{n-1-k+i}\right)}_{\hat{\zeta}_i^x}
\end{align*}
After some algebra, we get: 
\begin{align*}
\zeta_i^x- \hat{\zeta}_i^x &= \left(p_{i-1}-p_{i-2}\right) \frac{1- \sum\limits_{k=i-1}^n x_k}{1-x_i-x_{i-1}} + \sum\limits_{k=i+1}^n \frac{x_k}{1 - x_i}\left(p_{n-k+i}-p_{n-1-k+i} \right)\\
&+p_{i-1}\frac{x_{i-1} \sum\limits_{k=i+1}^{n}x_k}{(1-x_i)(1-x_i-x_{i-1})}- \sum\limits_{k=i+1}^n \frac{x_{i-1}x_k}{(1-x_i)(1-x_i-x_{i-1})}p_{n-1-k+i}.
\end{align*}
Simplifying further, we get: 
\begin{align*}
\zeta_i^x- \hat{\zeta}_i &= \left(p_{i-1}-p_{i-2}\right) \frac{1- \sum\limits_{k=i-1}^n x_k}{1-x_i-x_{i-1}} + \sum\limits_{k=i+1}^n \frac{x_k}{1 - x_i}\left(p_{n-k+i}-p_{n-1-k+i} \right)\\
& +\sum\limits_{k=i+1}^n \frac{x_{i-1}x_k}{(1-x_i)(1-x_i-x_{i-1})}\left( p_{i-1}-p_{n-1-k+i}\right).
\end{align*}
However, the term on the second line of this equality is negative as $p_{i-1} \leq p_{n-1-k+i}$ for each $k \geq i+1$. Then we can write 
\begin{align*}
\zeta_i^x- \hat{\zeta}_i^x & \leq \left(p_{i-1}-p_{i-2}\right) \frac{1- \sum\limits_{k=i-1}^n x_k}{1-x_i-x_{i-1}} + \sum\limits_{k=i+1}^n \frac{x_k}{1 - x_i}\left(p_{n-k+i}-p_{n-1-k+i} \right).
\end{align*}
The right hand side of this inequality is a weighted average of incremental increase in the probabilities, and the sum of the weights is less than 1. Then by using the fact that $ p_n-p_{n-1}$ is the largest incremental increase in the probability, we get:
\begin{align*}
\zeta_i^x- \hat{\zeta}_i^x & \leq p_n- p_{n-1}.
\end{align*}
Therefore, condition~\eqref{eqn2} is satisfied, and worker $i$ finds it optimal to shirk upon observing $e_{i-1}=0$. Therefore, the trigger strategy profile comprises an equilibrium under compensation scheme $w^x$. 

\medskip

\noindent \underline{Step 2.} 
Moreover, we know that $w^x$ is the least costly compensation scheme to support the trigger strategy profile as an equilibrium. This is because $w^x$ satisfies condition~\eqref{eqn1}, a necessary condition to support the trigger strategy as an equilibrium, with equality. Any compensation scheme with a lower expected cost to the principal would violate this condition.

\medskip
\noindent \underline{Step 3.} 
Finally, in any alternative equilibrium strategy profile with joint effort as the equilibrium outcome, the probability of success following worker $i$'s deviation (from his perspective) is (weakly) higher than that under the trigger strategy profile, $\zeta_i^x$. This imposes a (weakly) stronger version of the above constraint and thus requires a (weakly) higher payment to each worker. Consequently, $w^x$ must also be the optimal compensation scheme conditional on replacement strategy $x$. 
\end{proof}

\bigskip
\begin{proof}[\textbf{Proof of Proposition~\ref{prop: randomization}}] \hfill \\
Consider a replacement strategy $x$ of the form: 
$$x = (\rho,0,\ldots,0,1-\rho), \ \text{ for $\rho \in [0,1]$}.$$
The two extreme points within this class, which are obtained by setting $\rho=0$ and $\rho=1$, correspond to the optimal pure replacement strategies. We will show that choosing any interior value $\rho \in (0,1)$ outperforms these pure strategies. Therefore, the optimal replacement strategy necessarily involves randomization.

With a slight abuse of notation, for $x = (\rho, 0, \ldots, 0, 1-\rho)$, we denote $\zeta_i^x$ as ${\zeta}_i^\rho$, $w_i^x$ as ${w}_i^\rho$, and $W_x$ as ${W}_\rho$.
It is clear that:
\begin{align*}
{\zeta}_1^\rho  &= p_1,\\
{\zeta}_n^\rho &=p_{n-1},\\
{\zeta}_i^\rho &=  \left(1-\rho\right) p_i+ \rho \, p_{i-1}, \  \forall i \in \{2, \ldots, n-1\}.
\end{align*} 
Moreover,  
\[
\frac{{W}_\rho}{c}=\sum_{i=2}^{n-1} \frac{ p_n}{ p_n-\left(1-\rho\right) p_i-\rho p_{i-1}}+\frac{ p_n}{ p_n-p_1}\left(1-\rho\right)+\rho \frac{ p_n}{ p_n-p_{n-1}}+1.
\] Note this holds even when $\rho=0$ or $\rho=1$. The second order derivative of the $W_\rho$ with respect to $\rho$  satisfies:

\[
\frac{\partial^2  {W}_\rho}{\partial \rho^2} \, \frac{1}{c} =2 p_n\left(\sum_{i=2}^{n-1}\frac{\left(p_i-p_{i-1}\right)^2}{\left( p_n-\left(1-\rho\right) p_i-\rho p_{i-1}\right)^3} \right).
\] 
With $n \geq 3$, the summation contains at least one strictly positive term, as $\left(p_i - p_{i-1}\right)^2 > 0$ and $\left(p_n - \left(1-\rho\right) p_i - \rho p_{i-1}\right)^3 > 0$. Therefore:
$$ \frac{{W}_\rho}{\partial \rho^2} > 0, $$
implying that  $W_\rho$ is strictly convex in $\rho$. But we know that $\rho=0$ and $\rho=1$ corresponds to the optimal pure replacement strategies. That is, ${W}_\rho \big|_{\rho=0}= {W}_\rho \big|_{\rho=1}$. Then, from the convexity of $W_\rho$,  there must exist 
a $\rho \in (0,1)$ minimizing ${W}_\rho$. This implies that, a replacement strategy with some interior $\rho$ is better than the optimal pure replacement strategies. Therefore, the optimal scheme must involve randomization.  
\end{proof}

\bigskip 
\begin{proof}[\textbf{Proof of Proposition~\ref{prop: x2=0}}] \hfill \\
\noindent \underline{\underline{Step 1.}}
First we prove a stronger statement that implies $x_2^*=0$ when $n=3$.

\begin{claim*}
 In the optimal replacement strategy, $x_{n-1}^*=0$.   
\end{claim*} 

\noindent \textit{Proof of the Claim.} Let $x=(x_1,\cdots,  x_n)$ be an optimal replacement strategy, and suppose that $x_{n-1}>0$ to get a contradiction. Then 
we must also have $x_{n-1}<1$ as the optimal replacement strategy involves randomization (Proposition~\ref{prop: randomization}).

Recall that the expected compensation cost for the principal satisfies:
\[
\frac{W_{x}}{c}=\sum_{i=1}^n x_i+\sum_{i=1}^n \frac{\left(1-x_i\right)^2 p_n}{R_i},
\] where
\[
R_i=\left(1-x_i\right)\left( p_n-p_{i-1}\right)-\sum_{k=i+1}^n x_k\left(p_{n+i-k}-p_{i-1}\right).
\]
Now, given this optimal replacement strategy with $x_{n-1} > 0$, consider a deviation in which the probability of replacement for the $(n-1)$th worker is transferred to the $n$th worker. That is, the new replacement strategy $x'$ obtained by this deviation from $x$ is:
\[
x' = (x_1, \cdots, x_{n-2}, 0, x_{n-1} + x_n).
\]

\noindent \underline{Case 1:} First consider the case where $x_{n-1}+x_n<1$. The new expected cost is:
\[
\frac{W_{x'}}{c}=\sum_{i=1}^n x_i+\sum_{i=1}^{n-2} \frac{\left(1-x_i\right)^2 p_n}{R_i'}+\frac{p_n}{R_{n-1}'}+\frac{\left(1-x_{n-1}-x_n\right)^2 p_n}{R_{n}'},
\] where
\begin{align*}
R_i'&=R_i+x_{n-1}\left(p_{i+1}-p_i\right), \text{ if } i<n-1,\\
R_{n-1}'&=R_{n-1}+x_{n-1}\left(p_n-p_{n-1}\right),\\
R_n&=\left(1-x_n\right)\left(p_n-p_{n-1}\right),\\
R_n'& =\left(1-x_{n-1}-x_n\right)\left(p_n-p_{n-1}\right).
\end{align*}
We would like to show that $W_{x'}< W_x$ holds. Given $R_i < R_i'$ if $i<n-1$, it suffices to show that:
\begin{equation}
\frac{\left(1-x_n\right)^2}{R_n}-\frac{\left(1-x_{n-1}-x_n\right)^2}{R_n'}+\frac{\left(1-x_{n-1}\right)^2}{R_{n-1}}-\frac{1}{R_{n-1}'}\geq 0.
\label{forward}
\end{equation} 
The first two terms of \eqref{forward} can be rewritten as:
\begin{align}
\label{firsttwo}
\begin{split}
\frac{\left(1-x_n\right)^2}{R_n}-\frac{\left(1-x_{n-1}-x_n\right)^2}{R_n'}
=&\frac{\left(1-x_n\right)^2}{\left(1-x_n\right)\left(p_n-p_{n-1}\right)}-\frac{\left(1-x_{n-1}-x_n\right)^2}{\left(1-x_{n-1}-x_n\right)\left(p_n-p_{n-1}\right)}\\
=&\frac{1-x_n}{p_n-p_{n-1}}-\frac{1-x_{n-1}-x_n}{p_n-p_{n-1}}\\
=&\frac{x_{n-1}}{p_n-p_{n-1}}
\end{split}
\end{align}
Moreover, by using the equality $R_{n-1}=R_{n-1}'-x_{n-1}\left(p_n-p_{n-1}\right)$, we can rewrite the third and fourth terms of \eqref{forward} as:
\begin{align*}
\frac{\left(1-x_{n-1}\right)^2}{R_{n-1}}-\frac{1}{R_{n-1}'}&=\frac{R_{n-1}'\left(1-x_{n-1}\right)^2-R_{n-1}'+x_{n-1}\left(p_n-p_{n-1}\right)}{R_{n-1} R_{n-1}'}\\
&=\frac{x_{n-1}\left[\left(p_n-p_{n-1}\right)-R_{n-1}'\left(2-x_{n-1}\right)\right]}{R_{n-1} R_{n-1}'}.
\end{align*}
Therefore, inequality~\eqref{forward} holds if and only if
\[
\frac{1}{p_n - p_{n-1}} + \frac{(p_n - p_{n-1}) - R_{n-1}'(2 - x_{n-1})}{R_{n-1} R_{n-1}'} \geq 0.
\]
This inequality holds if and only if
\[
(p_n - p_{n-1})^2 - R_{n-1}'(2 - x_{n-1})(p_n - p_{n-1}) + R_{n-1} R_{n-1}' \geq 0,
\]
which further simplifies to
\[
(p_n - p_{n-1})^2 - R_{n-1}'(2 - x_{n-1})(p_n - p_{n-1}) + \left[R_{n-1}' - x_{n-1}(p_n - p_{n-1})\right] R_{n-1}' \geq 0.
\]
This can be rewritten as
\[
(p_n - p_{n-1})^2 - 2(p_n - p_{n-1}) R_{n-1}' + (R_{n-1}')^2 \geq 0,
\]
which is equivalent to
\[
\left(p_n - p_{n-1} - R_{n-1}'\right)^2 \geq 0.
\]
Since the square of a real number is always non-negative, the inequality holds. 

This implies that the deviation is profitable for the principal, which contradicts the optimality of the allocation \( x = (x_1, \dots, x_{n-1}, x_n) \). Therefore, the optimal replacement strategy $x^*$ has to satisfy \( x_{n-1}^* = 0 \).

\noindent \underline{Case 2:} Next consider the case where $x_{n-1}+x_n=1$. The first two terms in \eqref{forward}
becomes $\frac{\left(1-x_n\right)^2}{R_n}$ in this case. Moreover, note that:
\[
\frac{\left(1-x_n\right)^2}{R_n}=\frac{\left(1-x_n\right)^2}{\left(1-x_n\right)\left(p_n-p_{n-1}\right)}=\frac{1-x_n}{p_n-p_{n-1}}=\frac{x_{n-1}}{p_n-p_{n-1}} 
\] which coincides with \eqref{firsttwo}. Therefore, the optimal replacement strategy $x^*$ has to satisfy \( x_{n-1}^* = 0 \). This completes the proof of the claim. 

\medskip 

\noindent \underline{\underline{Step 2.}}
Next, we prove that the optimal replacement strategy $x^*$ satisfies \( x_3^* > 0 \). Suppose, for contradiction, that a replacement strategy $x= (x_1,x_2,x_3)$ with \( x_3 = 0 \) is optimal.
Note that, the partial derivative of the principal's expected cost with respect to \( x_1 \) (as long as $x_1<1$) is:
\begin{align}
\frac{1}{c} \frac{\partial W_{x}}{\partial x_1} 
&= 1 - p_3 \frac{w_1}{c} + p_3 (1 - x_1) \frac{1}{c} \frac{\partial w_1}{\partial x_1} \nonumber \\
&= 1 - \frac{p_3}{p_3 - \zeta_1} + p_3 (1 - x_1) \frac{1}{(p_3 - \zeta_1)^2} \frac{\partial \zeta_1}{\partial x_1} \nonumber \\
&= -\frac{\zeta_1}{p_3 - \zeta_1} + p_3 \frac{1}{(p_3 - \zeta_1)^2} \frac{\sum_{k=2}^3 x_k (p_{4-k} - p_0)}{1 - x_1} \nonumber \\
&= -\frac{\zeta_1}{p_3 - \zeta_1} + \frac{p_3 (\zeta_1 - p_0)}{(p_3 - \zeta_1)^2} \nonumber \\
&= \frac{\zeta_1^2 - p_3 p_0}{(p_3 - \zeta_1)^2}.
\label{eq: derivative1}
\end{align}
From the previous step, we know that \( x_2 = 0 \). Then \( x_2 = x_3 = 0 \) implies that \( \zeta_1 = p_0 \). By substituting this into \eqref{eq: derivative1}, we get  $\frac{\partial W_{x}}{\partial x_1} \leq 0.$  
This, in turn, implies that \( (x_1,x_2,x_3) = (1,0,0) \) is an optimal replacement strategy, which contradicts Proposition~\ref{prop: randomization}. Therefore, the optimal replacement strategy $x^*$ satisfies $x_3^*>0$.

\medskip 

\noindent \underline{\underline{Step 3.}}
Now we prove that the optimal replacement strategy $x^* $ satisfies \( x_3^* \geq x_1^* \). Suppose not for contradiction. Then, a replacement strategy $x= (x_1,0,x_3)$ with \( x_1 > x_3 \) is optimal.  This in turn requires $x_1 \in (0,1)$. 
 
The expected compensation cost of the principal under the scheme $\left(x_1, 0, x_3\right)$ is:
\[
c\left(x_1+x_3+p_3\left[\frac{(1-x_1)^2}{(1-x_1)(p_3-p_0)- x_3 (p_1-p_0)} 
+\frac{1}{(p_3-p_1)- x_3 (p_2-p_1)}
+\frac{1-x_3}{p_3-p_2}\right]\right).
\] 
The expected compensation cost of the principal under the scheme 
$\left(x_3, 0, x_1\right)$ is: 
\[
c\left(x_1+x_3+p_3\left[\frac{(1-x_3)^2}{(1-x_3)(p_3-p_0)- x_1 (p_1-p_0)} 
+\frac{1}{(p_3-p_1)- x_1 (p_2-p_1)}
+\frac{1-x_1}{p_3-p_2}\right]\right).
\] 
In the subsequent analysis, we will show that the precedent is strictly larger than the latter, i.e., we will show that:
\begin{align}
&\left[\frac{(1-x_1)^2}{(1-x_1)(p_3-p_0)- x_3 (p_1-p_0)} +\frac{1}{(p_3-p_1)- x_3 (p_2-p_1)}+\frac{1-x_3}{p_3-p_2}\right] \nonumber 
\\&-\left[\frac{(1-x_3)^2}{(1-x_3)(p_3-p_0)- x_1 (p_1-p_0)} +\frac{1}{(p_3-p_1)- x_1 (p_2-p_1)}+\frac{1-x_1}{p_3-p_2}\right]>0. \label{eqn6}
\end{align}
This, in turn, implies that the principal is better off under the allocation \((x_3, 0, x_1)\), which contradicts the optimality of \((x_1, 0, x_3)\).

To this end, note first that:
\[
\frac{1-x_3}{p_3-p_2}-\frac{1-x_1}{p_3-p_2}=\frac{x_1-x_3}{p_3-p_2}.
\]
Moreover, we have:
\begin{align*}
\frac{1}{(p_3-p_1)- x_3 (p_2-p_1)}-\frac{1}{(p_3-p_1)- x_1 (p_2-p_1)}
&=\frac{-\left(p_2-p_1\right)\left(x_1-x_3\right)}{\left[(p_3-p_1)- x_3 (p_2-p_1)\right]\left[(p_3-p_1)- x_1 (p_2-p_1)\right]}.
\end{align*} 
\begin{claim*}
The following inequality holds:
\begin{align*}
\frac{-\left(p_2-p_1\right)\left(x_1-x_3\right)}{\left[(p_3-p_1)- x_3 (p_2-p_1)\right]\left[(p_3-p_1)- x_1 (p_2-p_1)\right]}>\frac{-\left(p_2-p_1\right)\left(x_1-x_3\right)}{(p_3-p_1)(p_3-p_2)}. 
\end{align*}
\end{claim*}
\begin{proof}[Proof of the Claim]
The inequality holds if and only if:
\[
\left[(p_3-p_1)- x_3 (p_2-p_1)\right]\left[(p_3-p_1)- x_1 (p_2-p_1)\right]>(p_3-p_1)(p_3-p_2).
\] This is equivalent to:
\[
\left(p_3-p_1\right)^2-x_1\left(p_2-p_1\right)\left(p_3-p_1\right)-x_3\left(p_2-p_1\right)\left(p_3-p_1\right)+x_1 x_3\left(p_2-p_1\right)^2-(p_3-p_1)(p_3-p_2)>0,
\] 
which can be rewritten as:
\[
\left(p_3-p_1\right)\left(p_2-p_1\right)-x_1\left(p_2-p_1\right)\left(p_3-p_1\right)-x_3\left(p_2-p_1\right)\left(p_3-p_1\right)+x_1 x_3\left(p_2-p_1\right)^2>0.
\] 
Rearranging further, we obtain:
\[
\left(p_3-p_1\right)\left(1-x_1-x_3\right)+x_1 x_3\left(p_2-p_1\right)>0,
\] 
which is correct because $1-x_1-x_3 \geq 0$, $x_1>0$, and $x_3>0$.     
\end{proof}

Getting back to the inequality~\eqref{eqn6}, note that:
\[
\frac{x_1-x_3}{p_3-p_2}-\frac{\left(p_2-p_1\right)\left(x_1-x_3\right)}{(p_3-p_1)(p_3-p_2)}=\frac{x_1-x_3}{p_3-p_1}.
\] 
Define:
\begin{gather*}
\frac{(1-x_1)^2}{(1-x_1)(p_3-p_0)- x_3 (p_1-p_0)}-\frac{(1-x_3)^2}{(1-x_3)(p_3-p_0)- x_1 (p_1-p_0)}
=\frac{\Omega}{\Theta},
\end{gather*} where
\begin{align*}
\Theta &=\left[(1-x_1)(p_3-p_0)- x_3 (p_1-p_0)\right]\left[(1-x_3)(p_3-p_0)- x_1 (p_1-p_0)\right]   \\[0.75em] 
\Omega &=\left(p_3-p_0\right)\left[\left(1-x_1\right)^2\left(1-x_3\right)-\left(1-x_3\right)^2\left(1-x_1\right)\right]+\left(p_1-p_0\right)\left[\left(1-x_3\right)^2 x_3-\left(1-x_1\right)^2 x_1\right]\\
&=-\left(x_1-x_3\right)\left[\left(p_3-p_0\right)\left(1-x_1\right)\left(1-x_3\right)+\left(p_1-p_0\right)\left(x_1^2+x_1 x_3 +x_3^2-2x_1-2x_3+1\right)\right].
\end{align*}
Also denote
\[
\Lambda \equiv -\frac{\Omega}{x_1-x_3}.
\] 
To show that inequality~\eqref{eqn6} holds, it suffices to show:
\[
\frac{x_1-x_3}{p_3-p_1}-\frac{\left(x_1-x_3\right)\Lambda}{\Theta} \geq 0.
\] 
This is equivalent to:
\[
\Theta-\left(p_3-p_1\right)\Lambda \geq 0.
\]
By rearranging it, we obtain:
\begin{align*}
&\underbrace{\left(1-x_1\right)\left(1-x_3\right)\left(p_3-p_0\right)^2}_{S_1}-\underbrace{\left(p_3-p_1\right)\left(p_3-p_0\right)\left(1-x_1\right)\left(1-x_3\right)}_{S_2}\\
&-\underbrace{\left(1-x_1\right)x_1\left(p_3-p_0\right)\left(p_1-p_0\right)}_{S_3}-\underbrace{x_3\left(1-x_3\right)\left(p_1-p_0\right)\left(p_3-p_0\right)}_{S_4}\\
&-\underbrace{\left(p_3-p_1\right)\left(p_1-p_0\right)\left(x_1^2+x_1 x_3 +x_3^2-2x_1-2x_3+1\right)}_{S_5}+\underbrace{x_1 x_3\left(p_1-p_0\right)^2}_{S_6 }\geq 0.
\end{align*} 
However,  note that:
\begin{align*}
S_1-S_2=\left(p_1-p_0\right)\left(p_3-p_0\right)\left(1-x_1\right)\left(1-x_3\right).
\end{align*}
Thus, we obtain:
\begin{align*}
\underbrace{\left(p_1-p_0\right)\left(p_3-p_0\right)\left(1-x_1\right)\left(1-x_3\right)}_{S_1-S_2}-S_3 -S_4 
=\left(p_3-p_0\right)\left(p_1-p_0\right)\left(x_1^2+x_1 x_3+x_3^2-2x_1-2x_3+1\right).
\end{align*}
Therefore, we further derive:
\begin{align*}
\underbrace{\left(p_3-p_0\right)\left(p_1-p_0\right)\left(x_1^2+x_1 x_3+x_3^2-2x_1-2x_3+1\right)}_{S_1-S_2-S_3-S_4} -S_5 +S_6
=\left(1-x_1-x_3\right)^2 \left(p_1-p_0\right)^2 \geq 0.
\end{align*}
Thus, inequality~\eqref{eqn6} holds, which contradicts with the optimality \((x_3, 0, x_1)\). Therefore, the optimal replacement strategy $x^*$ satisfies  $x_3^*\geq x_1^*$.\end{proof}

\bigskip
\begin{proof}[\textbf{Proof of Proposition~\ref{prop: useup}}]\hfill

\noindent We prove Proposition~\ref{prop: useup} in two steps: 
\begin{enumerate}
    \item If $p_1^2 - p_3 p_0 \leq 0$, the principal fully utilizes the AI capacity, i.e., $\bar{x}^* = 1$.
    \item If $p_1^2 - p_3 p_0 > 0$, the principal does not fully utilize the AI capacity, i.e.,  $\bar{x}^* < 1$.
\end{enumerate}

\noindent \underline{Step 1.} Assume that $p_1^2 - p_3 p_0 \leq 0$. Suppose, for contradiction, that the principal does not fully utilize the AI capacity under the optimal replacement strategy, i.e.,  a strategy $x$ with $\bar{x} < 1$ is optimal. 
Note that
\[
\zeta_1 = p_0 + \sum_{k=2}^3 \frac{x_k}{1-x_1} \left(p_{4-k} - p_0\right).
\]
As we know from Proposition~\ref{prop: x2=0} that $x_2 = 0$,  we obtain: 
\[
\zeta_1 = p_0 + \frac{x_3}{1-x_1} \left(p_1 - p_0\right).
\]
Since $\bar{x} < 1$, we must have $x_3 < 1-x_1$, so $\zeta_1 < p_1$. Therefore, we must have: 
\[
\zeta_1^2 - p_3 p_0 < p_1^2 - p_3 p_0 \leq 0.
\]
Then, from equation~\eqref{eq: derivative1}, we know that:  
\[
\frac{1}{c} \frac{\partial W_x}{\partial x_1} = \frac{\zeta_1^2 - p_3 p_0}{(p_3 - \zeta_1)^2} \quad \implies \quad \frac{\partial W_x}{\partial x_1} < 0.
\]
Therefore, increasing $x_1$ slightly, which is feasible since $x_1 + x_3 < 1$, leads to an improvement in terms of the expected cost of compensation. This contradicts the optimality of $x$. Hence, under the optimal replacement strategy $x^*$, the principal fully utilizes the AI capacity, i.e., $\bar{x}^*=1$.

\medskip 

\noindent \underline{Step 2.} Assume $p_1^2 - p_3 p_0 > 0$. Suppose, for contradiction, that the principal fully utilizes the AI capacity under the optimal replacement strategy, i.e., a replacement strategy $x$ with $\bar{x}=1$ is optimal. Then we must have $x_1 + x_3 = 1$, since Proposition~\ref{prop: x2=0} shows that $x_2 = 0$.
Note that, 
$x_3 =1-x_1$ implies:
\[
\zeta_1 = p_0 + \frac{x_3}{1-x_1} \left(p_1 - p_0\right) = p_1.
\]
By substituting $\zeta_1 = p_1$, we obtain:
\[
\zeta_1^2 - p_3 p_0 = p_1^2 - p_3 p_0 > 0.
\] 
Then, from equation~\eqref{eq: derivative1}, we know that: 
\[
\frac{1}{c} \frac{\partial W_x}{\partial x_1} = \frac{\zeta_1^2 - p_3 p_0}{(p_3 - \zeta_1)^2} \quad \implies \quad \frac{\partial W_x}{\partial x_1} > 0.
\]
This implies that decreasing $x_1$ slightly, which is feasible as $x_1 > 0$, leads to an improvement in terms of the expected compensation cost. This contradicts the optimality of $x$. Hence,  under the optimal replacement strategy $x^*$, the principal under-utilizes the AI capacity, i.e., $\bar{x}^*<1$, in this case.
\end{proof}

\bigskip
\begin{proof}[\textbf{Proof of Proposition~\ref{prop: monotonic_wage}}]\hfill

\noindent Recall that:
\begin{align*}
\zeta_1^x &=\, p_2 \frac{x_2}{1 - x_1} \, +\, p_1 \frac{x_3}{1 - x_1} \, + \, p_0\frac{1-x_1-x_2-x_3}{1-x_1},\\
\zeta_2^x&=  p_2 \frac{x_3}{1 - x_2}  \, + \, p_1 \frac{1-x_2-x_3}{1-x_2},\\
\zeta_3^x&=  p_2.
\end{align*}
\begin{enumerate}
\item 
From earlier results, we know that $x_3^* \in \left(0,1\right)$, and $x_2^*=0$. First, $x_2^*=0$ implies $\zeta^{x^*}_1\leq p_1$. Next, $x_3^* \in \left(0,1\right)$ implies $\zeta^{x^*}_2 \in \left(p_1,p_2\right)$.
In consequence, we have $\zeta^{x^*}_3>\zeta^{x^*}_2>\zeta^{x^*}_1$, which implies that the wages are increasing in the position after the optimal automation: $w^{x^*}_3>w^{x^*}_2>w^{x^*}_1$.


\item As $x_3^*>0$, we have $\zeta^{x^*}_1>p_0=\zeta^0_1$. This implies that 
the wage of the front-most worker increases after the automation: $w^{x^*}_1>w^0_1$.
\item $x_3^*>0$ also implies that $\zeta^{x^*}_2>p_1=\zeta^0_2$. Therefore, 
the wage of the middle worker also increases after the automation: $w^{x^*}_2>w^0_2$.
\item Finally, note that, regardless of the replacement strategy, we have: $\zeta^{x^*}_3= p_2$. Therefore, the wage of the end-most worker remains the same after the automation: $w^{x^*}_3=w^0_3$.
\end{enumerate}
\end{proof}

\bigskip 
\begin{proof}[\textbf{Proof of Proposition~\ref{prop: monotonic_paoff}}] \hfill

\noindent From Proposition~\ref{prop: monotonic_wage}, we know that after AI replacement, the wages of the front-most and middle workers increase, while the wage of the end-most worker remains unchanged. Additionally, the end-most worker faces a risk of replacement, whereas the middle worker does not. Consequently, the middle worker's payoff increases, while the end-most worker's payoff decreases relative to their payoffs prior to AI replacement. The front-most worker's payoff depends on his replacement probability—while his wage increases, his payoff may rise or fall.      
\end{proof}

\subsection{Proofs of Section~\ref{sec: complementarity}}

\begin{proof}[\textbf{Proof of Proposition~\ref{closedform}}]\hfill 

\noindent \underline{Step 1.} First, we will derive the closed-form solution for the optimal replacement strategy. Note that this functional form satisfies the necessary and sufficient condition for full utilization given in Proposition~\ref{prop: useup}, ensuring $\bar{x}^* = 1$. Moreover, Proposition~\ref{prop: x2=0} establishes that $x_2^* = 0$. Therefore, we have $x_1^* + x_3^* = 1$.

As the optimal replacement strategy must involve randomization (Proposition~\ref{prop: randomization}), we must have $x_1^*, x_3^* \in (0,1)$. Consider an arbitrary such strategy given by $x = (x_1, x_2, x_3) = (\rho, 0, 1-\rho)$ for some $\rho \in (0,1)$. Denoting, with a slight abuse of notation, the principal's expected compensation cost under this strategy by $W_\rho$, we obtain:
\[
\frac{W_\rho}{c} = 1 + \frac{p_3}{p_3 - p_1} (1 - \rho) + \frac{p_3}{p_3 - (1 - \rho) p_2 - \rho p_1} + \rho \frac{p_3}{p_3 - p_2}.
\]
Under the optimal replacement strategy, $\rho$ must satisfy:
\[
\frac{1}{c} \frac{\partial W_\rho}{\partial \rho}
= -\frac{p_3 (p_2 - p_1)}{\left(p_3 - (1 - \rho) p_2 - \rho p_1\right)^2} - \frac{p_3}{p_3 - p_1} + \frac{p_3}{p_3 - p_2} = 0.
\]
Rearranging, this simplifies to:
\[
\frac{p_2 - p_1}{\left(p_3 - (1 - \rho) p_2 - \rho p_1\right)^2} = \frac{1}{p_3 - p_2} - \frac{1}{p_3 - p_1} = \frac{p_2 - p_1}{(p_3 - p_2)(p_3 - p_1)}.
\]
Solving for $\rho$ yields:
\[
p_3 - (1 - \rho) p_2 - \rho p_1 = \sqrt{(p_3 - p_2)(p_3 - p_1)}.
\]
From which we obtain:
\[
\rho = \frac{\sqrt{p_3 - p_2} \left(\sqrt{p_3 - p_1} - \sqrt{p_3 - p_2}\right)}{p_2 - p_1}
= \frac{\sqrt{1 - \alpha} \left(\sqrt{1 - \alpha^2} - \sqrt{1 - \alpha}\right)}{\alpha - \alpha^2}
= \frac{\sqrt{1 + \alpha} - 1}{\alpha}.
\]
This provides the closed-form solution for the optimal replacement strategy as stated.
 
\medskip 

\noindent \underline{Step 2.} Next, we will show the comparative statics results. From the previous step, for a given $\alpha$, the optimal replacement probability of the front-most worker satisfies:
\[
x_1^* =  \frac{\sqrt{1+\alpha}-1}{\alpha}.
\]
Taking the derivative of $x_1^*$ with respect to $\alpha$, we obtain:  
\[
\frac{\partial x_1^*}{\partial \alpha} = \frac{\frac{\alpha}{2\sqrt{1+\alpha}} - \sqrt{1+\alpha} + 1}{\alpha^2}.
\]
The numerator of this expression is a decreasing function of $\alpha$ as:
$$\frac{\partial \left(\frac{\alpha}{2\sqrt{1+\alpha}} - \sqrt{1+\alpha} + 1 \right)}{\partial \alpha} = \frac{\sqrt{1+\alpha} - \frac{\alpha}{2} \frac{1}{\sqrt{1+\alpha}}}{1+\alpha} - \frac{1}{\sqrt{1+\alpha}} = -\frac{\alpha}{2(1+\alpha)^{3/2}} <0.$$
Moreover, the numerator evaluates to zero at $\alpha = 0$, which implies that it remains strictly negative for all $\alpha \in (0,1)$.  
Thus, we obtain:  
$$\frac{\partial x_1^*}{\partial \alpha} < 0,$$
proving that as the degree of complementarity increases ($\alpha$ becomes smaller), the likelihood of the front-most
(end-most) worker being replaced increases (decreases).
\end{proof}

\subsection{Proofs of Section~\ref{sec:TaskBased}}

\begin{proof}[\textbf{Proof of Proposition~\ref{prop:TaskBasedCompensation}}] \hfill 

\noindent By applying analogous arguments from the proof of Proposition~\ref{prop: compensation}, we can establish that under a given replacement strategy $x$:  
(i) the optimal compensation scheme sustains effort by inducing a trigger strategy equilibrium, and  
(ii) workers remain indifferent between exerting effort and shirking along the equilibrium path.  
For conciseness, we do not replicate these arguments here. Denote by $\zeta_i^x$ the probability of project success when worker $i$ shirks, assuming all other workers follow the trigger strategy profile. Worker $i$'s incentive constraint then becomes:
$$
p_n (1 - x_i) w_i - (1 - x_i) c \geq \zeta_i^x w_i.
$$
Under the optimal compensation scheme, this constraint binds, which yields:
$$
w_i^x = \frac{c}{p_n - \zeta_i^x}.
$$
To understand the expression $\zeta_i^x$ in the proposition statement, consider what happens when worker $i$ chooses to shirk. Only a fraction $1 - x_i$ of his tasks are unperformed. The overall project success rate then depends on how this shirking decision affects downstream behavior through peer monitoring.  
\begin{itemize}
\item With probability $x_i$, the AI completes worker $i$'s tasks, so the worker-AI unit contributes to the project. In this case, the project succeeds with probability $p_n$.
  
\item With probability $(1 - x_i) x_{i+1}$, worker $i$'s shirking is noticed by worker $i+1$, who responds by shirking. However, the AI assigned to worker $i+1$ completes his share of tasks, so the unit at $i+1$ still contributes to the project. The project thus succeeds with probability $p_{n-1}$.
  
\item With probability $(1 - x_i)(1 - x_{i+1}) x_{i+2}$, shirking cascades further. Worker $i+1$ shirks,  and $1-x_{i+1}$ fraction of the tasks performed by AI doesn't fully compensate, and worker $i+2$ observes this and also shirks. Yet, the AI at $i+2$ performs $x_{i+2}$ of the tasks, allowing partial contribution. The project succeeds with probability $p_{n-2}$.
  
\item This process continues recursively, with each successive worker potentially shirking in response, depending on how much of the previous worker’s task is conducted by AI. Eventually, the entire downstream chain might shirk, in which case no further contribution comes from any of them.
  
\item In the worst case, where all successors also shirk and none of their AI compensates, the project succeeds only with probability $p_{i-1}$, which reflects the effort of workers before position $i$.
\end{itemize}
Summing over all these possible outcomes gives us:
$$
\zeta_i^x=x_i p_n+\sum_{k=1}^{n-i} \left[ p_{n-k} \, x_{i+k}\prod_{j=i}^{i+k-1}\left(1-x_j\right) \right]  + p_{i-1}\prod_{j=i}^n\left(1-x_j\right).
$$
This expression represents the expected project success probability when worker $i$ shirks, accounting for the compensating effects of AI task completion and the induced monitoring response throughout the network.
\end{proof}

\bigskip
\begin{proof}[\textbf{Proof of Proposition~\ref{prop:TaskBasedOptimal}}] \hfill 

\noindent For the case where $x_i < 1$ for each worker $i \in \{1, \dots, n\}$, we can write:
\[
p_n-\zeta_i^x=\left(1-x_i\right)\left(p_n-\Psi_i^x\right)
\] where 
$$ 
\Psi_i^x=x_{i+1}\,p_{n-1}+\sum_{k=1}^{n-i-1} \left[ x_{i+1+k} \, p_{n-1-k}\prod_{j=i+1}^{i+k}\left(1-x_j\right)  \right]+ p_{i-1}\prod_{j=i+1}^n\left(1-x_j\right).
$$
Moreover, we can rewrite the principal's expected cost as:
\begin{align*}
\sum\limits_{i=1}^n \left[ x_i c + p_n (1 - x_i) w_i \right]
=&\sum\limits_{i=1}^n \left[ x_i c + p_n (1 - x_i) \frac{c}{p_n-\zeta_i^x} \right]\\
=&\sum\limits_{i=1}^n \left[ x_i c + p_n \, \frac{c}{p_n-\Psi_i^x} \right].
\end{align*}
First, note that $\Psi_i^x$ depends on $x_\ell$ only if $\ell > i$. Moreover, the partial derivative $\frac{\partial \Psi_i^x}{\partial x_\ell}$ is proportional to the difference between $p_{n+i-\ell}$ and a weighted average of $p_{n+i-\ell-1}, \dots, p_{i-1}$, and is therefore positive.
Now consider increasing $x_i$. The principal's expected compensation cost  increases due to the following strictly increasing components:
\begin{itemize}
    \item the direct term $x_i c$, and
    \item the indirect terms of the form $p_n \frac{c}{p_n - \Psi_\ell^x}$ for $\ell < i$, each of which increases because $\Psi_\ell^x$ is strictly increasing in $x_i$.
\end{itemize}
Therefore, increasing $x_i$ strictly increases the expected cost. It follows that setting all $x_i = 0$ minimizes the expected cost over all strategies satisfying $x_i < 1$ for each worker $i \in \{1, \dots, n\}$. This proves that the strategy $(0, \dots, 0)$ is optimal among such replacement strategies. Lemma~\ref{topbottom} completes the proof.
\end{proof}

\subsection{Proofs of Section~\ref{sec: network}}

\begin{proof}[\textbf{Proof of Proposition~\ref{prop: star_optimum}}]\hfill

\noindent 
First, similar reasoning to that in Proposition~\ref{prop: compensation} establishes, for a given replacement strategy, the optimal compensation scheme, and that this scheme induces a trigger strategy as an equilibrium.  
Moreover, analogous arguments to those in Lemma~\ref{lem:existence} ensure the existence of an optimal replacement strategy in this star network structure.   
To prove the proposition, we first state and prove two auxiliary claims.

\begin{claim}\label{evenly}
Given a fixed $x_n$, there exists a constrained optimal replacement strategy in which $$x_1 = x_2=\cdots = x_{n-1}.$$    
\end{claim}
\begin{proof}[Proof of Claim~\ref{evenly}] \hfill

\noindent -- If $x_n=1$, the claim trivially holds as we must have $x_1=\cdots=x_{n-1}=0$ in this case.

\noindent -- If $x_n=0$, then $\zeta_i^x=p_{n-2}$, for each $i\in \{1, \cdots, n-1\}$.  The expected compensation cost of the principal under an arbitrary replacement strategy $(x_1,x_2, \cdots, x_{n-1},0)$ is: 
\[
\sum_{i=1}^{n-1} \left[x_i c+ \left(1-x_i\right)p_n \frac{c}{p_n-p_{n-2}}\right]+ p_n\frac{c}{p_n-p_{n-1}}.
\] 
Consider modifying this replacement strategy into $\left(\tilde{x}, \cdots, \tilde{x},0\right)$ with:
\[
\tilde{x}=\frac{x_1+\cdots+x_{n-1}}{n-1}.
\] 
The expected compensation cost of the principal under this replacement strategy is:
\[
\left(n-1\right)\left[\tilde{x} c+ \left(1-\tilde{x}\right)p_n \frac{c}{p_n-p_{n-2}}\right]+ p_n\frac{c}{p_n-p_{n-1}}.
\]
Note that, this is equal to the expected compensation cost prior to modification. This implies that there exists a constrained optimal replacement strategy satisfying $x_1=\cdots=x_{n-1}$.

\noindent -- If $x_n \in (0,1)$, we first show that $\left(1-x_i\right) p_n \frac{w_i^x}{c}$ is a strictly convex function of $x_i$. Note that:
\begin{align*}
\frac{\partial \left(  \left(1-x_i\right) p_n \frac{w_i^x}{c}\right)}{\partial x_i} &= p_n\left( -\frac{w_i^x}{c}+\frac{1}{c}\left(1-x_i\right)\frac{\partial w_i^x}{\partial x_i} \right) \\
&=p_n\left(-\frac{1}{p_n-\zeta_i^x}+\frac{1-x_i}{\left(p_n-\zeta_i^x\right)^2}\frac{\partial \zeta_i^x}{\partial x_i}\right)\\
&=p_n\left(-\frac{1}{p_n-\zeta_i^x}+\frac{1}{\left(p_n-\zeta_i^x\right)^2}\frac{x_n\left(p_{n-1}-p_{n-2}\right)}{1-x_i}\right)\\
&=p_n\left(-\frac{1}{p_n-\zeta_i^x}+\frac{\zeta_i^x-p_{n-2}}{\left(p_n-\zeta_i^x\right)^2}\right)\\
&=p_n\left(\frac{2\zeta_i^x-p_n-p_{n-2}}{\left(p_n-\zeta_i^x\right)^2}\right).
\end{align*}
Taking the second order derivative, we obtain:
\begin{align*}
\frac{\partial^2 \left(  \left(1-x_i\right) p_n \frac{w_i^x}{c}\right)}{\partial x_i^2} &= p_n \, \frac{\partial \zeta_i^x}{\partial x_i} \, \frac{2\left(p_n-\zeta_i^x\right)^2+2\left(p_n-\zeta_i^x\right)\left(2\zeta_i^x-p_n-p_{n-2}\right)}{\left(p_n-\zeta_i^x\right)^4}\\
&=p_n \, \frac{\partial \zeta_i^x}{\partial x_i} \,\frac{2\left(p_n-\zeta_i^x\right)\left(\zeta_i^x-p_{n-2}\right)}{\left(p_n-\zeta_i^x\right)^4}.
\end{align*}
But, $x_n>0$ implies $\zeta_i^x-p_{n-2}>0$ and $\frac{\partial \zeta_i^x}{\partial x_i}>0$. Also, note that $p_n-\zeta_i^x>0$. Therefore, the second order derivative of $\left(1-x_i\right) p_n \frac{w_i^x}{c}$ with respect to $x_i$ is positive, implying that  $\left(1-x_i\right) p_n \frac{w_i^x}{c}$ is a strictly convex function of $x_i$.

Now, denoting
$$\Gamma\left(x_i\right) \equiv \frac{\left(1-x_i\right) p_n c}{ p_n-p_{n-2}\left(1-\frac{x_n}{1-x_i}\right)-p_{n-1} \frac{x_n}{1-x_i}},$$
we can write the expected compensation cost of the principal under the replacement strategy $\left(x_1, \cdots, x_n\right)$ as:
\begin{align*}
\sum_{k\neq i, j} \left[x_k c+\left(1-x_k\right) p_n w_k^x\right]
+x_i c+x_j c+\Gamma\left(x_i\right)+\Gamma\left(x_j\right).
\end{align*} 
Consider the following optimization problem:
\begin{align*}
\min_{\substack{{x_i'\geq 0}\\ {x_j'\geq 0}}} \quad x_i' c+x_j' c+\Gamma\left(x_i'\right)+\Gamma\left(x_j'\right)\quad 
\text{s.t. } &x_i'+x_j'=x_i+x_j.
\end{align*}
We can equivalently write this optimization problem as:
\begin{align*}
\min_{x_i'} \quad \Gamma\left(x_i'\right)+\Gamma\left(x_i+x_j-x_i'\right) \quad  \text{s.t. } 0 \leq x_i' \leq x_i+x_j.
\end{align*} 
Note that:
\[
\frac{\partial^2 \left(\Gamma\left(x_i'\right)+\Gamma\left(x_i+x_j-x_i'\right)\right)}{\partial x_i'^2}=\Gamma''\left(x_i'\right)+\Gamma''\left(x_i+x_j-x_i'\right)>0.
\] 
Thus, the first-order condition  
\[
\frac{\partial \left(\Gamma\left(x_i'\right)+\Gamma\left(x_i+x_j-x_i'\right)\right)}{\partial x_i'} =\Gamma'\left(x_i'\right)-\Gamma'\left(x_i+x_j-x_i'\right)=0
\]  
is uniquely satisfied by \( x_i' = \frac{x_i + x_j}{2} \), ensuring that the minimum is uniquely achieved. This implies that for all \( i, j \in \{1, \dots, n-1\} \), we must have \( x_i = x_j \) in the optimal replacement strategy. Consequently, it follows that \( x_1 = \cdots = x_{n-1} \).  
\end{proof}

\begin{claim}
The principal chooses to fully utilize AI capacity, setting $\bar{x}=1$, if the condition $p_{n-1}^2-p_n p_{n-2}\leq 0$ is satisfied. 
\label{useup_star}
\end{claim}

\begin{proof}[Proof of Claim~\ref{useup_star}] \hfill

\noindent Suppose, for contradiction, that the condition $p_{n-1}^2-p_n p_{n-2}\leq 0$ is satisfied, yet the principal does not fully utilize the AI capacity.  
From the previous claim, we know that for a given replacement probability $x_n$ of the central worker, the rest of the optimal replacement strategy can, without loss of generality, be restricted to strategies where the replacement probability is identical for all the peripheral workers, i.e., $x_i = x$ for each worker $i \leq  n-1$, whereby $x < \frac{1 - x_n}{n-1}$.  

Take such a replacement strategy $(x,x,\cdots, x, x_n)$ for some $x < \frac{1-x_n}{n-1}$. 
Ignoring the dependence on the specific values of \( x \) and \( x_n \), let \( w \) denote the corresponding wage of a peripheral worker and let \( \zeta \) represent a peripheral worker's belief about the project's success rate when deviating under the trigger strategy profile. These are given by:  
\[
\zeta = p_{n-2} + \frac{x_n}{1-x} (p_{n-1} - p_{n-2}), \quad  w =  \frac{c}{p_n - \zeta}.
\]
The principal's expected cost for a peripheral node is:  
\[
x c + (1 - x) p_n w.
\]
Substituting the expression for \( w \) and differentiating with respect to \( x \), we obtain:  
\begin{align*}
\frac{\partial \left( x c + (1 - x) p_n w \right)}{\partial x} 
&= c\left(  1 - \frac{p_n}{p_n - \zeta} + \frac{(1 - x) p_n}{(p_n - \zeta)^2} \frac{\partial \zeta}{\partial x} \right)\\
&= c\left( -\frac{\zeta}{p_n - \zeta} + \frac{p_n}{(p_n - \zeta)^2} \frac{x_n (p_{n-1} - p_{n-2})}{1 - x} \right)\\
&= c\left( -\frac{\zeta}{p_n - \zeta} + \frac{p_n (\zeta - p_{n-2})}{(p_n - \zeta)^2} \right)\\
&= c\left( \frac{\zeta^2 - p_n p_{n-2}}{(p_n - \zeta)^2} \right).
\end{align*}
Since \( \zeta < p_{n-1} \), the condition \( p_{n-1}^2 - p_n p_{n-2} \leq 0 \) ensures that the derivative above is negative.  
Thus, increasing \( x \) slightly---which is feasible since the capacity is not fully utilized---leads to a reduction in expected compensation cost.  
This contradicts optimality, implying that the principal must fully utilize the AI capacity. \end{proof}

Claims~\ref{evenly} and \ref{useup_star} imply that given \( x_n \), the optimal replacement strategy satisfies:  
\[
x_i = x = \frac{1 - x_n}{n-1}, \quad \forall i \leq n-1.
\]  
The principal's expected compensation cost for a given \( x_n \), denoted by \( W_{x_n} \) (with a slight abuse of notation), satisfies:  
\[
\frac{W_{x_n}}{c} =
1 + (n-1)(1-x) p_n \frac{1}{p_n - \zeta} + (1 - x_n) p_n \frac{1}{p_n - p_{n-1}}.
\]
Differentiating with respect to \( x_n \), we obtain:  
{\small \begin{align*}
\frac{\partial \frac{W_{x_n}}{c \, p_n}}{\partial x_n} &=   
-\left(n-1\right) \frac{1}{p_n - \zeta} \frac{\partial x}{\partial x_n}  
+ \left(n-1\right) (1 - x) \frac{1}{(p_n - \zeta)^2} \frac{\partial \zeta}{\partial x_n}  
- \frac{1}{p_n - p_{n-1}} \\
&\quad - (n-1) \frac{1}{p_n - \zeta} \frac{\partial x}{\partial x_n}  
+ (n-1) (1 - x) \frac{1}{(p_n - \zeta)^2}  
\left[ \frac{p_{n-1} - p_{n-2}}{1 - x}  
+ \frac{(p_{n-1} - p_{n-2}) x_n}{(1 - x)^2} \frac{\partial x}{\partial x_n} \right]  
- \frac{1}{p_n - p_{n-1}}.
\end{align*}}
After simplifications, and defining \( R \equiv   p_n - \zeta \), this reduces to:  
\[
\frac{\partial \frac{W_{x_n}}{c \, p_n}}{\partial x_n} = \frac{1}{R} + \frac{(n-1) (p_{n-1} - p_{n-2})}{R^2} - \frac{(p_{n-1} - p_{n-2}) x_n}{(1 - x) R^2} - \frac{1}{p_n - p_{n-1}}.
\]
Therefore, the sign of \( \frac{\partial \frac{W_{x_n}}{c \, p_n}}{\partial x_n} \) matches the sign of:  
\[
\left(p_n - p_{n-1}\right) (1 - x) R  
+ (n-1) (p_n - p_{n-1}) (p_{n-1} - p_{n-2}) (1 - x)  
- (p_{n-1} - p_{n-2}) (p_n - p_{n-1}) x_n  
- (1 - x) R^2,
\]
which in turn has the same sign as:  
\[
-\left(1 - x\right) R^2  
+ \left(p_n - p_{n-1}\right) (1 - x) R  
+ (n-2) (p_{n-1} - p_{n-2}) (p_n - p_{n-1}).
\]
Moreover, note that  
\[
1 - x = \frac{(n-2)(p_{n-1} - p_{n-2})}{R + (n-1)(p_{n-1} - p_{n-2}) - (p_n - p_{n-2})}.
\]  
Thus, the sign of the above expression matches the sign of:  
\begin{align*}
-(n-2)(p_{n-1} - p_{n-2}) R^2  
+ 2(n-2)(p_{n-1} - p_{n-2})(p_n - p_{n-1}) R \\
+ (n-2)(p_{n-1} - p_{n-2})(p_n - p_{n-1})  
\left[(n-1)(p_{n-1} - p_{n-2}) - (p_n - p_{n-2})\right].
\end{align*}  
We know that \( R \in [p_n - p_{n-1}, p_n - p_{n-2}] \) and  that the expression above attains its minimum at \( R = p_n - p_{n-2} \). Evaluating at \( R = p_n - p_{n-2} \), we get:  
\begin{align*}
-(n-2)(p_{n-1} - p_{n-2})(p_n - p_{n-2})^2  
+ 2(n-2)(p_{n-1} - p_{n-2})(p_n - p_{n-1})(p_n - p_{n-2}) \\
+ (n-2)(p_{n-1} - p_{n-2})(p_n - p_{n-1})  
\left[(n-1)(p_{n-1} - p_{n-2}) - (p_n - p_{n-2})\right].
\end{align*}  
Dividing this expression by \( (n-2)(p_{n-1} - p_{n-2}) \), we obtain:  
\[
(p_{n-1} - p_{n-2})\left[(n-2)(p_n - p_{n-1}) - (p_{n-1} - p_{n-2})\right],
\]  
which is positive.  

Therefore, the derivative of the expected cost with respect to \( x_n \) is positive.  
As a result, the optimal strategy requires \( x_n = 0 \). Furthermore, as shown in the proof of Claim~\ref{evenly}, all replacement strategies satisfying \( \sum_{i=1}^{n-1} x_i = 1 \) are optimal.  
\end{proof}

\section{Complementary Analysis} 
\label{sec: A_complementary}
This appendix section provides a complementary analysis to the main findings in the document. In what follows, we use \emph{O-ring} production function and, for clarity and conciseness, continue with the setting of three workers.

\paragraph{AI's impact on intra-team hierarchy of payoffs.} Proposition~\ref{prop: order_paoff} presents additional results on how the optimal AI replacement affects the payoff hierarchy, showing that the middle worker receives the highest payoffs considering both wages and the frequency and the use of AI.

\bigskip 
\begin{proposition}
\label{prop: order_paoff}
The optimal AI replacement reshapes the hierarchy of payoffs within the team: the middle worker receives the highest payoff, followed by the end-most worker, and then the front-most worker. Formally:  
$$\left(1-x_2^*\right)\left(p_3w^{x^*}_2-c\right)>\left(1-x_3^*\right)\left(p_3w^{x^*}_3-c\right)>\left(1-x_1^*\right)\left(p_3w^{x^*}_1-c\right).$$
\end{proposition}
\bigskip

\begin{proof}[\textbf{Proof of Proposition~\ref{prop: order_paoff}}]\hfill

\noindent Denoting $\beta\equiv \sqrt{1+\alpha}$, we get:
\begin{align*}
\frac{\left(1-x_1^*\right)\left(p_3w^{x^*}_1-c\right)}{c} &=   \frac{\left(\beta^2-1\right)^2}{\beta\left(\beta+1\right)\left(2-\beta^2\right)} \\
\frac{\left(1-x_2^*\right)\left(p_3w^{x^*}_2-c\right)}{c} &=\frac{1}{\beta\left(2-\beta^2\right)}-1 \\
\frac{\left(1-x_3^*\right)\left(p_3 w^{x^*}_3-c\right)}{c} &=\frac{\beta-1}{2-\beta^2}.
\end{align*} 
Therefore:
\[
\frac{\left(1-x_2^*\right)\left(p_3w^{x^*}_2-c\right)}{c}-\frac{\left(1-x_3^*\right)\left(p_3 w^{x^*}_3-c\right)}{c}=\frac{1}{\beta\left(2-\beta^2\right)}-1-\frac{\beta-1}{2-\beta^2}=\frac{\left(\beta+1\right)\left(\beta-1\right)^2}{\beta\left(2-\beta^2\right)}>0,
\] 
and
\[
\frac{\left(1-x_3^*\right)\left(p_3w^{x^*}_3-c\right)}{c}-\frac{\left(1-x_1^*\right)\left(p_3 w^{x^*}_1-c\right)}{c}=\frac{\beta-1}{2-\beta^2}-\frac{\left(\beta^2-1\right)^2}{\beta\left(\beta+1\right)\left(2-\beta^2\right)}=\frac{\left(\beta^2-1\right)\left(-\beta^2+\beta+1\right)}{\beta\left(\beta+1\right)\left(2-\beta^2\right)},
\] which is positive as $\beta\in\left(1, \sqrt{2}\right)$. This concludes the proof.
\end{proof}

\bigskip
\paragraph{AI's impact on payoffs.} Proposition \ref{prop: worker1_paoff} demonstrates that the front-most worker may experience a loss or gain in payoff, and provides the conditions for each outcome under the O-ring production function. Recall that, in the main text, we have already demonstrated that the optimal AI adoption always resulted in higher (lower) payoffs for the middle (end-most) worker. This proposition therefore completes the earlier payoff analysis. 

\bigskip
\begin{proposition}
\label{prop: worker1_paoff}
The front-most worker’s payoff declines following AI replacement,
\[
(1 - x_1^*) \left(p_3 w_1^{x^*} - c \right) < p_3 w_1^0 - c,
\]
if and only if effort complementarity is weak enough---that is, \( \alpha \geq \bar{\alpha} \) for some threshold \( \bar{\alpha} \in (0,1) \).
\end{proposition}
\bigskip

\begin{proof}[\textbf{Proof of Proposition~\ref{prop: worker1_paoff}}]\hfill

\noindent Note that:
{\small
{\begin{align*}
\frac{p_3w^0_1-c}{c}- &\frac{\left(1-x_1^*\right)\left(p_3w^{x^*}_1-c\right)}{c}
=\frac{\left(\beta^2-1\right)^3}{\left(2-\beta^2\right)\left[\left(\beta^2-1\right)^2+\beta^2\right]}-\frac{\left(\beta^2-1\right)^2}{\beta\left(\beta+1\right)\left(2-\beta^2\right)}\\
&=\frac{\left(\beta^2-1\right)^2}{\beta\left(\beta+1\right)\left(2-\beta^2\right)\left[\left(\beta^2-1\right)^2+\beta^2\right]}\left(\beta\left(\beta+1\right)\left(\beta^2-1\right)-\left[\left(\beta^2-1\right)^2+\beta^2\right]\right).
\end{align*}}}
The second line is a positive value multiplied with $\beta^3-\beta-1$,
which is positive if and only if $\beta$ is larger than $\sqrt[3]{\frac{1}{2}+\sqrt{\frac{23}{108}}}+\sqrt[3]{\frac{1}{2}-\sqrt{\frac{23}{108}}}\approx 1.3247$ (by Cardano's formula).
\end{proof}

As is indicated in Proposition~\ref{prop: monotonic_paoff}, the optimal AI deployment on the one hand increases the wage of the front-most worker. On the other hand, he faces replacement. Proposition~\ref{prop: worker1_paoff} implies that the latter effect dominates when $\alpha$ is large (tasks are highly independent of each other).

\bigskip

\paragraph{AI's impact on intra-team inequality of payoffs.} 
Proposition~\ref{prop: payoff_inequality} explores how AI alters intra-team payoff inequality and shows that, similarly to wage inequality, the inequality of payoffs also declines with optimal AI adoption.

\bigskip 
\begin{proposition} \label{prop: payoff_inequality}
The intra-team payoff inequality decreases with optimal AI adoption, relative to the intra-team inequality without AI adoption. That is: 
\[
\left[\left(1-x_2^*\right)\left(p_3w^{x^*}_2-c\right)-\left(1-x_1^*\right)\left(p_3 w^{x^*}_1-c\right)\right]
<\left[\left(p_3w^0_3-c\right)-\left(p_3 w^0_1-c\right)\right].
\]
\end{proposition}
\bigskip
\begin{proof}[\textbf{Proof of Proposition~\ref{prop: payoff_inequality}}]\hfill

\noindent Note that:
\begin{align*}
\frac{1}{c} &\left[\left(1-x_2\right)\left(p_3w^x_2-c\right)-\left(1-x_1\right)\left(p_3 w^x_1-c\right)\right]
-\frac{1}{c}\left[\left(p_3w^0_3-c\right)-\left(p_3 w^0_1-c\right)\right]\\
&=\left[\frac{1}{\beta\left(2-\beta^2\right)}-1-\frac{\left(\beta^2-1\right)^2}{\beta\left(\beta+1\right)\left(2-\beta^2\right)}\right]-\left[\frac{1}{2-\beta^2}-\frac{1}{\left(2-\beta^2\right)\left[\left(\beta^2-1\right)^2+\beta^2\right]}\right]\\
&=\frac{\beta-1}{2-\beta^2}-\frac{\beta^2\left(\beta^2-1\right)}{\left(2-\beta^2\right)\left[\left(\beta^2-1\right)^2+\beta^2\right]}\\
&=\frac{\beta-1}{\left(2-\beta^2\right)\left[\left(\beta^2-1\right)^2+\beta^2\right]}\left\{\left(\beta-1\right)\left[\left(\beta-1\right)\left(\beta^2+\beta-1\right)-3\right]-1\right\}
\end{align*} which is negative as $\beta\in\left(1, \sqrt{2}\right)$.
This completes the proof.
\end{proof}
